\documentclass[letterpaper, DIV=18]{scrartcl}

\usepackage{amssymb, amsthm, mathtools,bbm,bm}
\usepackage[square,sort,comma,numbers]{natbib}
\usepackage{enumitem}

\usepackage[none]{hyphenat}
\newtheorem{theorem}{Theorem}
\newtheorem{corollary}[theorem]{Corollary}
\newtheorem{lemma}[theorem]{Lemma}

\newtheorem{proposition}[theorem]{Proposition}
\newtheorem{assumption}[theorem]{Assumption}
\numberwithin{equation}{section}
\numberwithin{theorem}{section}
\theoremstyle{remark}
\newtheorem{remark}[theorem]{Remark}
\newtheorem{remarks}[theorem]{Remarks}

\newcommand{\bxi}{\boldsymbol{\xi} }

\newcommand{\supp}{{\operatorname{supp}\,}}
\newcommand{\tr}{{\operatorname{Tr}\,}}

\newcommand{\indfct}{\operatorname{1}}

\newcommand{\beq}[1]{\begin{equation} \label{#1}}
	\newcommand{\eeq}{\end{equation}}

\renewcommand{\epsilon}{\varepsilon}

\newcommand{\spa}{\operatorname{span}}
\newcommand{\essup}{\operatorname{esssup}}

\DeclareMathOperator{\dom}{dom}

\def\be{\begin{equation}}
	\def\ee{\end{equation}}

\begin{document}
	\addtokomafont{author}{\raggedright}
	
\title{ \raggedright  A Parisi Formula for Quantum Spin Glasses}

	\author{\hspace{-.075in} Chokri Manai and Simone Warzel}
	\date{\vspace{-.3in}}
	
	\maketitle

	%\begin{abstract}
	\minisec{Abstract}
	We establish three equivalent versions of a Parisi formula for the free energy of mean-field spin glasses in a transversal magnetic field. 
	These results are derived from available results for classical vector spin glasses by an approximation method using the functional integral representation of the partition function.  In this approach, the order parameter is a non-decreasing function with values in the non-negative, real hermitian Hilbert-Schmidt operators. 
	For the quantum Sherrington-Kirkpatrick model, we also show that under the assumption of self-averaging of the self-overlap, the optimising  Parisi order parameter is found within a two-dimensional subspace spanned by the self-overlap and the fully stationary overlap. 
	%\end{abstract}

\tableofcontents

\section{Setting the stage}
Transcending their initial roots in physics~\cite{Mezard:1986aa}, classical mean-field spin glasses, as the Sherrington Kirkpatrick model~\cite{SK75},  have found interest in mathematics  and computer science~\cite{Bov06,MM09,Tal11a,Tal11b,Cont12}.  Their widespread appeal is rooted in the universal applicability of the tools and concepts associated with replica symmetry breaking.    At their core, these models exhibit a phase transition into a glass phase, which is marked by the breaking of replica symmetry,  and described through Parisi's functional order parameter~\cite{Parisi:1979aa,Parisi:1980aa}. 
Part of the charm of these seemingly simple mean-field models is that  Parisi's  replica approach enables the derivation of a variational formula for the free energy~\cite{Mezard:1986aa,Tal06,Tal11a,Tal11b,Cont12,Pan13}.

Soon after Parisi's discovery~\cite{Parisi:1979aa,Parisi:1980aa}, physicists have started searching for an extension of Parisi's replica approach to quantum models~\cite{Sommers:1981aa,FS86}. More recently, the interest in quantum glasses have been renewed due to questions related to quantum information~\cite{TTC17,CLPTS23}, but also since disordered mean-field models still serve as paradigms in quantum matter or, as the SYK model, even in high-energy physics \cite{Georges:2000aa,Young:2017aa,Baldwin:2020aa}.

The quantum nature of the spin degrees of freedom may enter such models in various degrees of difficulty. In case the non-commutativity is present in the interaction terms, such as in SYK or a Heisenberg spin glass, an extension of Parisi's formalism is wide open, and very little is known rigorously (see e.g.~\cite{Morita:2005ab,Cont12,ES14}). In case the non-commutativity only enters in the form of a transversal field, the resulting model may be analysed using a functional integral. Based on this, replica symmetry breaking approximations to the free energy have been found early on~\cite{FS86,Yamamoto:1987aa,Usadel:1987aa} (see also~\cite{BSS96,SIC13} and references therein). However, aside from the case of hierarchical transversal models of generalised random energy type \cite{Goldschmidt:1990aa,MW20,MW21,MW22,Manai:2023ys}, a closed-form variational formula for the free energy has been absent. The purpose of this paper is to present a first version of such a variational formula of Parisi type for a large class of transversal field models, which includes the quantum Sherrington-Kirkpatrick (QSK) model. Using the functional integral representation of the partition function, we derive our three equivalent versions of such a formula from available formulae for vector-spin glasses \cite{Pan18,Chen23,Chen23b}. For the QSK, we thus complete an earlier approach~\cite{AdBr20}, in which a related variational expression has been obtained. The latter, however, suffers from the crucial drawback of an additional limit in its formulation. Our variational principle extends 
the variational principle for the free energy of the QSK in the high-temperature phase from~\cite{Leschke:2021aa} to the full parameter range. 

The variational principle will identify the quantum analogue of Parisi's functional order parameter with monotone functions taking values in a subset of the non-negative Hilbert-Schmidt operators. That subset matches  the support of  the functional integral. 
For the QSK, we also show how, under a widely believed assumption of the self-averaging of the self-overlap, our variational formula simplifies in the sense that the optimising Hilbert-Schmidt operators are from a two-parameter subspace spanned by the self-overlap and the stationary kernel. 
This simplification of the order parameter  is akin to similar results for other glasses with a symmetry in the underlying a-priori measure such as in case of balanced Potts glasses~\cite{BaSo23,Elderfield:1983}.\\

To present our main results, which are collected and discussed in Section~\ref{sec:main}, we first need to set the stage by introducing the models, the functional integral and the quantum Parisi functional, which involves several correlation functionals. This is the purpose of the remainder of this section.

\subsection{Models and assumptions} 
We consider systems of spin-$\frac{1}{2} $ quantum degrees of freedom. For each spin, the observable algebra on $ \mathbb{C}^2 $ is generated by the identity~$ \mathbbm{1} $ and the Pauli matrices:
$$
S^x = \left(\begin{matrix} 0 & 1 \\ 1 & 0 \end{matrix}\right), \quad S^y =\left( \begin{matrix} 0 & i \\ -i & 0 \end{matrix}\right), \quad S^z =\left( \begin{matrix}  1 & 0  \\ 0 & -1 \end{matrix}\right) .
$$
The Hilbert space of $ N $ spin-$\frac{1}{2} $ is the $ N $-fold tensor product $ \bigotimes_{j=1}^N \mathbb{C}^2 $. The Pauli matrices are lifted to the tensor product in the canonical way, i.e. for the $ j$th spin the non-trivial factor is inserted in the $ j$th tensor component, $ S_j^{x} \coloneqq \mathbbm{1} \otimes \dots \otimes S^{x} \otimes \dots \mathbbm{1} $ and similarly for $ y,z $. 

The quantum spin glasses studied in this paper are transversal field models, whose Hamiltonians on $ \bigotimes_{j=1}^N \mathbb{C}^2 $ are of the form
\begin{equation}\label{def:HN}
H_N \coloneqq U\left(S_1^z, \dots S_N^z\right) +  \beta h \sum_{j=1}^N S_j^z - \beta b \sum_{j=1}^N  S_j^x . 
\end{equation}
The last term represents a constant transversal magnetic field in the (negative) $ x $-direction. The second term is a constant  longitudinal field in the $ z $-direction, which agrees with the axis along which the interaction among the spins is mediated. The operator $U\left(S_1^z, \dots, S_N^z\right)$ is defined via the functional calculus of the jointly diagonalisable $z$-components of the spins in terms of a random function $U: \{-1,1\}^{N} \to \mathbb{R}$, stemming from a classical, mean-field Ising spin glass like the Sherrington-Kirkpatrick model. 
More generally, we assume $  U $ to be a realisation of a Gaussian random process with  mean zero, i.e.\ $ \mathbb{E}\left[U( \boldsymbol{\sigma} ) \right] = 0 $ for all $  \boldsymbol{\sigma} \in  \{-1,1\}^{N} $, and covariance given in terms of the classical overlap of two Ising configurations $  \boldsymbol{\sigma}^{(1)} , \boldsymbol{\sigma}^{(2)} \in \{-1,1\}^N $: 
\begin{equation}\label{eq:classcorr}
\mathbb{E}\left[U\big( \boldsymbol{\sigma}^{(1)} \big) \ U\big(\boldsymbol{\sigma}^{(1)} \big)  \right]  = N \  \widehat\zeta\left(  N^{-1} \boldsymbol{\sigma}^{(1)}  \cdot \boldsymbol{\sigma}^{(2)}  \right) , \qquad  \boldsymbol{\sigma}^{(1)}  \cdot \boldsymbol{\sigma}^{(2)}  \coloneqq \sum_{j=1}^N \sigma_j^{(1)} \sigma_j^{(2)} . 
\end{equation}
Throughout this paper, the function $ \widehat \zeta $, which determines the covariance, will be tacitly assumed to satisfy:
\begin{assumption}\label{ass}
The function $  \widehat \zeta: [-1,1] \to \mathbb{R} $ admits an absolutely convergent series representation of the form
\begin{equation}\label{def:zetahat}
\widehat\zeta(x) \coloneqq  \sum_{p=1}^\infty \beta_{2p}^2 x^{2p} , \qquad \mbox{with}\quad\sum_{p=1}^\infty (2p)^2\beta_{2p}^2 < \infty . 
\end{equation}
\end{assumption} 
This covers all classical, mixed, even $ p $-spin interactions and, in particular, the Sherrington-Kirkpatrick model, for which $ \beta_{2p} =  0 $  unless $ p = 1 $. More precisely, each of the terms in~\eqref{def:zetahat} represents via~\eqref{eq:classcorr}
the covariance of a $ p $-spin interaction on $\boldsymbol{\sigma} \coloneqq (\sigma_1, \dots , \sigma_N) $  of the form 
\begin{equation}\label{eq:pspinU}
U(\boldsymbol{\sigma}) = \frac{1}{N^{(p-1)/2}} \sum_{j_1, \dots , j_p = 1}^N g_{j_1,\dots , j_p} \sigma_{j_1} \cdots \sigma_{j_p}
\end{equation}
with $\left(g_{j_1,\dots , j_p} \right) $ independent and identically distributed, standard Gaussian random variables. Note that such potentials do include self-interactions, which however do not contribute to the free energy in the thermodynamic limit. We restrict to even $ p $-spin interactions due to issues with positivity and convexity.

The central quantity of interest is the quantum partition function $ \tr e^{-H_N } $ corresponding to~\eqref{def:HN}, and its free energy.  We choose to include 
the inverse temperature $ \beta \in [0,\infty) $ in the Hamiltonian. The interaction term incorporates the temperature through the factors  $ \beta_{2p} $ in~\eqref{def:zetahat}.

\subsection{Functional integral}
As in other works on transversal field models, our analysis of the quantum partition function is based on a functional integral representation~\cite{Chayes:2008aa,AdBr20,Leschke:2021aa}.  
For its formulation, 
we consider a Poisson process $ \omega_1 $ on $ [0,1) $ with intensity $ \beta b \geq 0 $ and conditioned to have an even number of points in $ [0,1) $. Setting $ \eta_1(t,\omega_1) := \#\{ \mbox{Poisson points of $ \omega_1 $ in $ [0,t) $} \} $ and picking initial conditions $ \sigma_1 \in \{-1,1\} $ independently with probability $ 1/2 $, we 
let $ \mu_1 $ stand for the  induced probability measure on the path 
$$ \xi_1 : [0,1) \to \{-1,1\} , \quad
\xi_1(t) := \sigma_1  \ (-1)^{\eta_1(t)} .
$$
Any such path is continuous from the right with a limit from the left ('c\`adl\`ag'). Moreover, since the number of jumps is even, by identifying $ \xi_1(1) := \xi_1(0) $, it can be uniquely extended to a c\`adl\`ag path on the unit torus. The measure $ \mu_1 $ is hence a well-defined probability measure on paths $ \xi_1 $ on the unit torus endowed with the  Borel sigma-algebra $ \mathcal{B} $  corresponding to the Skohorod topology. 

Since our Hamiltonian $ H_N $ also incorporates a constant longitudinal magnetic field, $ \beta h \in \mathbb{R}  $, we modify the apriori probability measure $ \mu_1 $ by setting for any $ A \in  \mathcal{B}  $:
\begin{equation}
	\nu_1(A )  := \frac{\cosh( \beta b) }{ \cosh( \beta  \sqrt{b^2 + h^2} ) } \int_A \exp\left( -\beta h \int_0^1 \xi_1(s) ds \right) \mu_1(d\xi_1) .
\end{equation}
The prefactor ensures that $ \nu_1 $ is again a probability measure. We now set $ \nu_N \coloneqq  \nu_1^{\otimes N} $ the product measure of $ N $ identical copies of $ \nu_1 $, which govern the collection of paths $ \bxi = (\xi_1, \dots , \xi_N ) $.  The partition function of $ H_N $,
$$
 \tr e^{- H_N}  = \left( 2 \cosh( \beta  \sqrt{b^2 + h^2} ) \right)^N Z_N ,
$$
 is known \cite{Leschke:2021aa} to be given by the functional integral
$$
Z_N \coloneqq \int \exp\left( -\int_0^1 U\left(\bxi(s)\right) ds \right) \nu_N(d\bxi) . 
$$
The main purpose of this paper is to derive a quantum Parisi formula for  $ Z_N $. Loosely speaking, this will be done by identifying $ Z_N $ as the partition function of a mixed $p $-spin model with continuous 'vector' spins $ \bxi = (\xi_1, \dots , \xi_N) $. The action  $ \int_0^1 U\left(\bxi(s)\right) ds  $ in the functional integral serves as its random energy. By~\eqref{eq:classcorr} the latter is a Gaussian process with mean zero and covariance
\begin{equation}\label{def:covariance}
\mathbb{E}\left[\int_0^1 U\left(\bxi^{(1)}(s)\right) ds \ \int_0^1 U\left(\bxi^{(2)}(t)\right) dt \right] = N  \int_0^1 \int_0^1 \widehat\zeta\left( N^{-1} \bxi^{(1)}(s) \cdot \bxi^{(2)}(t) \right) ds dt . 
\end{equation}
for two replicas $ \bxi^{(1)},  \bxi^{(2)} $ of the paths.

\subsection{Space of asymptotic overlap operators}
Given two replicas of paths $ \bxi^{(1)},  \bxi^{(2)} $, the covariance~\eqref{def:covariance} only depends on the 
path-overlap $$ \widehat r_N(\bxi^{(1)},  \bxi^{(2)} )(s,t) \coloneqq N^{-1} \bxi^{(1)}(s) \cdot \bxi^{(2)}(t) .  $$ 
Through its kernel the latter defines 
a rank-$N $ overlap operator $ r_N(\bxi^{(1)},  \bxi^{(2)} ) $ on the real Hilbert space 
$$ 
L^2(0,1) := \left\{ f: (0,1) \to \mathbb{R} \ \big| \ \int_0^1 f(t)^2 dt < \infty \right\} . 
$$
The real scalar product on this Hilbert space is denoted by $ \langle f, g \rangle := \int_0^1 f(s) g(s) ds $. 
Recall that an integral operator 
$ \varrho$ on $ L^2(0,1) $ with kernel $ \widehat \varrho: [0,1]^2 \to \mathbb{R} $ acts as
\begin{equation}\label{def:kern}
(\varrho f)(s) := \int_0^1 \widehat\varrho(s,t) f(t) dt , \quad f \in L^2(0,1) .
\end{equation} 
We will mostly be concerned with hermitian integral operators, whose kernels are real and symmetric by definition, i.e. $ \widehat\varrho(s,t)  = \widehat\varrho(t,s) $, and amongst those with Hilbert-Schmidt operators. They form the real Hilbert space  $ \mathcal{S}_2 $ with the scalar product $$ \langle \varrho , \varrho' \rangle := \tr  \varrho  \varrho' = \int_0^1 \int_0^1 \widehat\varrho(t,s) \widehat\varrho'(s,t)  \ ds dt . $$ 
Any real symmetric function $   \widehat\varrho \in L^2 (0,1)^2 $ uniquely defines though~\eqref{def:kern} 
an operator $\varrho \in \mathcal{S}_2 $ on $ L^2(0,1) $ and vice versa. Any such $  \varrho $  then also belongs to the real Banach space of hermitian Schatten-$p $ operators $ \mathcal{S}_p $ with $ p \geq 2 $. 
Generally,  $\varrho \in \mathcal{S}_2 $ is not necessarily in the trace class $ \mathcal{S}_1 $, and it is generally rather difficult to decide whether $\varrho \in \mathcal{S}_1 $ or not, solely based on the integral kernel $ \widehat \varrho $.
However, if $ \varrho $ belongs to the convex cone  $  \mathcal{S}_2^+ \coloneqq \{  \varrho \in  \mathcal{S}_2 \ | \   \varrho  \geq 0 \} $ of non-negative, hermitian Hilbert-Schmidt operators, 
a result of Weidmann~\cite{Weid66} ensures that if the kernel is essentially bounded, $ \| \widehat \varrho \|_\infty \coloneqq \essup_{s,t\in (0,1)^2}  |\widehat \varrho(s,t) | < \infty $, then $  \varrho  \in \mathcal{S}_1 $. 
More precisely, in this situation the function defined for almost every (a.e.)\ $ \tau \in (0,1) $ by
\begin{equation}\label{eq:regkern}
\widehat\varrho(\tau,\tau) \coloneqq \lim_{r\downarrow 0} \frac{1}{(2r)^2} \int_{-r}^r  \int_{-r}^r   \widehat\varrho(s+\tau, t + \tau) ds dt 
\end{equation}
is non-negative and bounded with $ \tr\varrho = \int_0^1  \widehat \varrho(\tau,\tau)  \ d \tau $ and $  |  \widehat\varrho(s,t)  | \leq \sqrt{ \widehat \varrho(s,s) \ \widehat\varrho(t,t) } $ 
for a.e.\ $ (s,t) \in (0,1)^2 $, cf.\ \cite{Weid66,Bris88} and Proposition~\ref{lem:propzetak2} in Appendix~\ref{app:HS}. 
Subsequently, we will tacitly assume that the diagonal of the bounded integral kernel of $ \varrho \in \mathcal{S}_2^+ $ is given by~\eqref{eq:regkern}.

In our spin glass model the self-overlap  $ r_N(\bxi,\bxi)$ is an example of a non-negative integral operator with a bounded kernel. Its kernel agrees with $ N^{-1} \bxi(s) \cdot \bxi(t)  $, and the diagonal is obtained by setting $ s= t $ by the Lebesgue differentiation theorem.  
The self-overlap hence belongs to the set
$$
\mathcal{D} \coloneqq \left\{ \varrho \geq 0 \ \mbox{is an integral operator on $ L^2(0,1) $ with some kernel $  \widehat \varrho: (0,1)^2 \to [-1,1] $} \right\} .
$$
This subset of $ \mathcal{S}_2^+$ contains the support,  $ \supp \nu_1$, of the functional integral. It is convex and closed, i.e.\ $ \mathcal{D} = \overline{\mathcal{D}} $, where here and in the following, we equip $ \mathcal{D} $ with the Hilbert-Schmidt norm. 
This set will serve as the domain of the quantum Parisi functional, which we construct in the next two subsections. 
\subsection{Correlation functional} 
As is apparent from~\eqref{def:covariance} the correlation functional $ \zeta: \mathcal{D} \to [0,\infty)$,
\begin{equation}\label{def:zeta} 
\zeta(\varrho) \coloneqq  \int_0^1 \int_0^1 \widehat\zeta\left(\widehat\varrho(s,t) \right) ds dt
\end{equation}
plays an important role. 
Since $ \widehat \zeta \geq 0 $, the integral in~\eqref{def:zeta}  is well defined, potentially infinite, for any measurable kernel $ \widehat \varrho $. It thus extends to the Hilbert space $\mathcal{S}_2 $. 
By the representation~\eqref{def:zetahat} of $ \widehat \zeta $, we may represent 
the correlation functional as a series involving the even Hadamard products of $ \varrho $:
\begin{equation}\label{eq:Hadamard}
 \zeta(\varrho) \coloneqq \sum_{p= 1}^\infty \beta_{2p}^2 \ \langle 1 , \varrho^{\odot 2p} 1 \rangle , \quad \mbox{with} \quad\langle 1 , \varrho^{\odot 2p} 1 \rangle \coloneqq  \int_0^1 \int_0^1 \widehat\varrho(s,t)^{2p} ds dt . 
 \end{equation}
The abbreviation $  \langle 1 , \varrho^{\odot 2p} 1 \rangle $ alludes to the fact that in case $ \widehat \varrho $ is bounded, this term agrees with the matrix element involving the constant function one,
$ 1 \in L^2(0,1) $,  of the $2p $-fold Hadamard product   $  \varrho^{\odot 2p}  $  of $ \varrho $ (see Appendix~\ref{app:HS}, which compiles basic properties of this product).   %
We summarize some properties of $ \zeta $, which will be crucial for the definition of the Parisi functional, in the following lemma. 
\begin{lemma}\label{lem:zetaprop} 
\begin{enumerate}
\item  $  \zeta: \mathcal{S}_2^+ \to [0,\infty]$ is convex, weakly lower semicontinuous and proper, i.e.\ with non-empty domain $ \dom \zeta \coloneqq \left\{ \varrho \in  \mathcal{S}_2^+ \ | \  \zeta(\varrho) < \infty \right\} \supset \mathcal{D} $. For any $ \varrho \in \mathcal{D} $, we have $ \zeta(\varrho)  \leq \widehat\zeta(1) $. 
\item  The functional   $ \zeta $  is monotone, i.e.\ if $0 \leq \pi \leq \varrho  \in \mathcal{D} $, then $  0 \leq \zeta(\pi)  \leq \zeta(\varrho) $.  
\item At every $ \varrho \in \mathcal{D} $ the functional  
 $ \zeta $ is Gateaux differentiable with derivative represented by the non-negative operator
%\begin{equation} 
%d\zeta(\varrho)\pi  \coloneqq   \sum_{p= 1}^\infty 2p \  \beta_{2p}^2 \int_0^1 \int_0^1\varrho(s,t)^{2p-1}\ \pi(s,t)  ds dt  = \langle d\zeta(\varrho) , \pi \rangle ,
%\end{equation}
%where the operator 
\begin{equation}\label{eq:covariance form}
d\zeta(\varrho) \coloneqq \sum_{p\geq 1}  2p \beta_{2p}^2 \  \varrho^{\odot 2p -1} 
\end{equation}
on $ L^2(0,1) $. The latter is uniformly bounded, $ \sup_{\varrho \in \mathcal{D}} \| d\zeta(\varrho) \| \leq  N_\zeta := \sum_{p\geq 1} 2p  \beta_{2p}^2 < \infty  $. 
As a function of $ \varrho $, this operator is monotone increasing, i.e. if $0\leq  \pi \leq \varrho \in \mathcal{D} $, then $0\leq  d\zeta(\pi)  \leq d\zeta(\varrho) $, and Lipschitz continouus,
\begin{equation}\label{eq:Lcdzeta}. 
\left\| d\zeta(\varrho)  - d\zeta(\pi) \right\| \leq L_\zeta \int_0^1 \int_0^1 \left|  \widehat\varrho(s,t) -  \widehat\pi(s,t)  \right| dsdt \leq   L_\zeta \left\| \varrho - \pi  \right\|_2  
\end{equation}
for any $ \varrho, \pi \in \mathcal{D} $ with $ L_\zeta \coloneqq \sum_{p\geq 1} 2p  (2p-1) \beta_{2p}^2 $. 
\end{enumerate}
\end{lemma}
\begin{proof}%[Proof of Lemma~\ref{lem:zetaprop}] 
1.~The convexity immediately follows from the convexity of the functions $ [0,\infty) \ni x \mapsto x^{2p} $ for any $ p \in \mathbb{N} $. 
The weak lower semicontinuity follows from Proposition~\ref{prop:extensionzeta}. 
The claimed bound on $\zeta $ for $ \varrho \in \mathcal{D} $  is immediate from the boundedness of the kernel  $| \widehat \varrho |\leq 1 $. \\
\noindent
2.~This follows straightforwardly from Lemma~\ref{lem:propzetak1}.\\
\noindent
3.~The Gateaux differentiability and its representation in terms of the operator~\eqref{eq:covariance form} is immediate. The boundedness with the explicit norm bounds of the operator $ d\zeta(\varrho) $ 
are also straightforward from its definition as a Hadamard-Schur product. So is its monotonicity by Lemma~\ref{lem:propzetak1}. 
The proof of~\eqref{eq:Lcdzeta} follows from the same elementary inequality as~\eqref{eq:Lcontzetak}. 
\end{proof}

Employing the positive part   $ \varrho_+ = \varrho \ 1_{[0,\infty)}(\varrho) $ of the hermitian  
Hilbert-Schmidt operator $ \varrho = \varrho_+ - \varrho_- $, one extends
 $ \zeta : \mathcal{S}_2 \to [ 0 , \infty) $, $ \zeta(\varrho) \coloneqq \zeta(\varrho_+) $. This functional remains weakly lower semicontinuous and convex, since $ (a+b)_+ \leq a_+ + b_+ $ for any $ a,b \in \mathcal{S}_2 $. 
 Its Legendre transform 
\begin{equation}\label{def:zeta*}
 \zeta^*: \mathcal{S}_2 \to [0,\infty], \quad \zeta^*(x) \coloneqq \sup_{\varrho \in \mathcal{S}_2}\big(  \langle x ,\varrho \rangle - \zeta(\varrho) \big) .
\end{equation}
again defines a proper, convex, weakly lower semicontinuous functional, whose Legendre transform agrees with $ \zeta = \left(\zeta^*\right)^* $ by the Fenchel-Moreau theorem~\cite{BauCom17}. 
The domain of this one-sided Legendre transform $ \zeta^* $  is $  \dom \zeta^* = \mathcal{S}_2^+ $. 
If $ x \in \mathcal{S}_2^+ $, it is easy to see that the supremum in~\eqref{def:zeta*} coincides with the supremum restricted to $ \varrho \in \mathcal{S}_2^+ $.
\bigskip
 
The quantum Parisi functional also involves the functional $ \theta:  \mathcal{D} \to [0,\infty) $, 
\begin{equation}\label{eq:thetadef}
	\theta(\varrho) := \langle d\zeta(\varrho), \varrho\rangle  - \zeta(\varrho) =  \int_0^1 \int_0^1 \sum_{p= 1}^\infty (2p-1) \beta_{2p}^2\ \varrho(s,t)^{2p}\  ds dt  = \sum_{p=1 }^\infty (2p-1) \beta_{2p}^2 \  \langle 1, \varrho^{\odot 2p} 1 \rangle .  
\end{equation}
Its properties are listed in the following
\begin{lemma}\label{lem:thetaprop}
\begin{enumerate}
\item $ \theta $ is convex and uniformly bounded, $ \displaystyle \sup_{\varrho \in \mathcal{D}}  \theta(\varrho)  \leq  \sum_{p= 1}^\infty (2p-1)  \beta_{2p}^2 \leq N_\zeta  < \infty  $.
\item $ \theta $ is monotone increasing, i.e. if $0 \leq \pi \leq \varrho \in \mathcal{D} $, then $ 0 \leq \theta(\pi)  \leq \theta(\varrho) $.
\item $ \theta $ is Lipschitz continuous when $ \mathcal{D} $ is equipped with the Hilbert-Schmidt norm:
\begin{equation}\label{eq:thetaLc}
\left| \theta(\varrho) - \theta(\pi) \right| \leq L_\zeta \int_0^1 \int_0^1 \left| \widehat\varrho(s,t) - \widehat\pi(s,t)  \right| dsdt \leq L_\zeta \ \| \varrho - \pi \|_2 . 
\end{equation}
\end{enumerate}
\end{lemma}
\begin{proof}
1.~The convexity again follows from the convexity of the functions $ [0,\infty) \ni x \mapsto x^{2p} $ for any $ p \in \mathbb{N} $. The uniform boundedness is immediate from the boundedness of the kernel  $| \widehat \varrho |\leq 1 $.\\
\noindent
2.~The representation~\eqref{eq:thetadef} as a non-negative series involving the Hadamard product such that Proposition~\ref{pop:monotonicity} yields the claim. \\
 3.~The proof of the Lipschitz continuity parallels the proof of~\eqref{eq:Lcdzeta}. 
\end{proof}

\subsection{Quantum Parisi functional} 
In a first step, the space of order parameters of the quantum glass will  be identified with monotone increasing, piecewise constant, and right-continuous with left limits (c\`adl\`ag) functions with values in $ \mathcal{D} $:
\begin{align*}
\Pi := \   \left\{ \pi: [0,1] \to \mathcal{D} \ \big| \  0 \leq \pi(s) \leq \pi(t) \; \mbox{for all $ s \leq t $ with }\; \pi(0) = 0 , \;  
 \mbox{and constant on} \right. \quad &\\
 \mbox{some partition $ 0 =: m_0 < m_1 < \dots < m_r < m_{r+1}:= 1$:} \qquad & \\
  \left.  \quad \pi = \sum_{k=1}^{r+1} \indfct_{[m_{k-1},m_k) } \pi(m_{k-1}) + \pi(m_{r+1}) \indfct_{\{1\}} 
 \right\} . &
\end{align*}
Note  that  $ \pi(m) \in  \mathcal{D} $ for any $ m\in [0,1] $ once we require $ \pi(1) \in \mathcal{D} $. \\
To a function $ \pi\in \Pi $ with $ r $-steps, i.e.\ a partition $ 0 =: m_0 < m_1 < \dots < m_r < m_{r+1}:=  1$ with ordered values $ 0= \pi(m_0) \leq  \pi(m_1) \dots \leq \pi(m_r) \leq  \pi(1)  \in \mathcal{D} $, we associate  $ j = 1,  \dots , r+1  $ independent Gaussian random vectors in $ w_j \in L^2(0,1) $. They are characterised by their zero mean and their covariance functional
\begin{equation}
\mathbb{E}\left[ \langle f , w_j \rangle \langle g , w_j \rangle \right]  = \langle f , \left( d\zeta(\pi(m_j)) - d\zeta(\pi(m_{j-1})) \right) g \rangle , \quad f,g \in L^2(0,1) . 
\end{equation}
Note that Lemma~\ref{lem:zetaprop} in particular safeguards that $  d\zeta(\pi(m_j)) - d\zeta(\pi(m_{j-1}))  \in \mathcal{S}_2^+ $ is a legitimate covariance operator of a Gaussian process in $ L^2(0,1) $, cf.~\cite{IbRo78}. To any hermitian Hilbert-Schmidt operator $ x \in \mathcal{S}_2 $, we then associate the random variable 
\begin{equation}\label{eq:defXr}
	X_{r}(\pi,x) := \ln \int \exp\left(\sum_{j=1}^{r+1}\langle \xi, w_j \rangle  + \langle \xi  , \left( x- \tfrac{1}{2} d\zeta(\pi(1)) \right)  \xi  \rangle \right) \nu_1(d\xi) . 
\end{equation}
Note that $ \exp\left(\sum_{j=1}^{r+1}\langle \xi, w_j \rangle  -  \tfrac{1}{2} \langle \xi  , d\zeta(\pi(1))   \xi  \rangle \right) $ is an exponential martingale with respect to the Gaussian average. 
For $1 \leq  j \leq  r +1$, we proceed inductively and set
\begin{equation}\label{eq:defXrec}
	X_{j-1}(\pi,x) := \frac{1}{m_j} \ln\mathbb{E}_{j}\left[ e^{m_j X_{j}(\pi,x) } \right]  ,
\end{equation}
where $ \mathbb{E}_{j} $ denotes the average with respect to $ w_j $.

The quantum Parisi functional $ \mathcal{P} : \Pi \times \mathcal{S}_2 \to \mathbb{R} $ is  then defined as the non-random end-point of this iterative procedure plus an integral:
\begin{equation}\label{eq:Parisifctl}
\mathcal{P}\left(\pi, x \right) \coloneqq X_0(\pi,x) + \frac{1}{2}\int_0^1 \theta(\pi(m)) dm . 
\end{equation}

\begin{remarks} 
\begin{enumerate}
\item Since the last random variable $ w_{r+1} $ comes in~\eqref{eq:defXrec} with weight $ m_{r+1} = 1 $, it may be integrated out. This results in a random variable $ X_r(\pi,x) $, which is of the same form as $ X_{r+1}(\pi,x) $ with $ \pi(1) $ replaced by $ \pi(m_r) $. One may hence always assume $ \pi(1) = \pi(m_r) $ without loss of generality. 
\item To analyse the Parisi functional, we also rely on a well-known equivalent representation using Ruelle probability cascades~\cite{Ruelle87,Pan13}.  For the  above partition with $ r $ steps associated to $ \pi \in \Pi $, let $ (\mu^\pi_\alpha)_{\alpha \in \mathbb{N}^{r} }$ denote the weights of the corresponding Ruelle probability cascade \cite{Pan13}. One then defines a family of $ L^2 $-valued, centred Gaussian processes $ (z_\alpha)_{\alpha \in \mathbb{N}^{r} }$ with covariance
\begin{equation}\label{eq:defz}
\mathbb{E}\left[ \langle f , z_\alpha \rangle \langle g , z_{\alpha'} \rangle \right]  = \langle f , d\zeta\left(\pi(m_{\alpha \wedge \alpha'} )  \right)g \rangle , \quad f,g \in L^2(0,1) . 
\end{equation}
where $ \alpha \wedge \alpha' \coloneqq \min\left\{  0\leq j \leq r  \ \big| \ \alpha_{j} = \alpha'_{j} \; \mbox{and} \;  \alpha_{j+1} \neq \alpha'_{j+1} \right\} $ and by convention $ \alpha \wedge \alpha' = 0 $ in case $ \alpha_1 \neq \alpha'_1 $ (see.~\cite[Sec. 2]{Pan13} for an explicit construction of these processes). 
Then the iterative construction of $ X_0(\pi,x) $ for $ \pi \in \Pi $ with $ r $ steps is known to be equivalent to the following expectation value
\begin{equation}\label{eq:PRuelle}
 X_0(\pi,x)= \mathbb{E}\left[ \ln \sum_{\alpha \in \mathbb{N}^{r} } \mu^\pi_\alpha \int \exp\left( \langle \xi ,  z_\alpha \rangle + \langle \xi , \left( x- \tfrac{1}{2} d\zeta(\pi(1) \right)  \xi  \rangle \right) \nu_1(d\xi) \right] .
 \end{equation}
 \end{enumerate}
\end{remarks}

Before we proceed, we collect some further properties of $ \mathcal{P} $. 
\begin{proposition}\label{lem:LcontP}
 There is some $ L \in (0,\infty ) $ such that for any $ x , x'  \in  \mathcal{S}_2  $ and $ \pi , \pi' \in \Pi$:
\begin{equation}\label{eq:LcontP}
\left| \mathcal{P}\left(\pi, x \right) -  \mathcal{P}\left(\pi', x' \right)\right| \leq L \left( \| x- x' \| +  \int_0^1 \| \pi(s) - \pi'(s) \|_2 \ ds \right) .
\end{equation}
\end{proposition}
The proof, which is by standard Gaussian interpolation~\cite{Tal11a,Pan13}, is found in Appendix~\ref{app:PF}. For fixed $ x \in \mathcal{S}_2 $, one can hence continuously extend $ \mathcal{P}\left(\pi, x \right) $ to the set 
$$   \left\{ \pi: [0,1] \to \mathcal{D} \; \mbox{c\`adl\`ag} \  \big| \  0 \leq \pi(s) \leq \pi(t) \leq \pi(1)  \; \mbox{for all}\;  s \leq t \;  \mbox{with}\; \pi(0) = 0   \right\} .
$$

For future purpose, we also record some elementary bounds as well as smoothness and related properties of the Parisi functional. For its formulation, we recall that $ \| x \| $ stands for the operator norm on $ L^2(0,1)$, and 
$ \left\| \widehat x \right\|_{q}  \coloneqq \left( \int_0^1\int_0^1 \left| \widehat x(s,t)\right|^{q} ds dt \right)^{1/q}  $,  $ q \in [1,\infty) $, 
is the norm of the integral kernel $ \widehat x $ corresponding to $ x $. 
\begin{proposition}\label{prop:elpropP}
\begin{enumerate}
\item For any $\pi \in  \Pi$, the functional $ \mathcal{S}_2 \ni x \mapsto \mathcal{P}(\pi,x) $ is convex and  bounded. If  $ x = x_+ - x_- \in \mathcal{S}_2 $, then:
\begin{align*}
- \frac{N_\zeta}{2} -  \min\left\{ \| x_- \| , \| \widehat x_- \|_1\right\} & \leq  X_0(\pi,x) \leq \min\left\{ \| x_+ \| , \| \widehat x_+ \|_1\right\} . 
\end{align*}
The functional is Gateaux differentiable at every $ x \in \mathcal{S}_2 $  with bounded derivative, $ \left\| d\mathcal{P}(\pi, x) \right\|_2 \leq 1 $. The second derivative is non-negative and also uniformly bounded, i.e.\ for any $ y \in \mathcal{S}_2 $:
\begin{equation} 
0 \leq \frac{d^2}{d\lambda^2} \mathcal{P}(\pi, x+ \lambda y ) \Big|_{\lambda = 0 } \leq 2 \left\| y \right\|_2^2 . 
\end{equation}
\item $ \displaystyle  \inf_{\pi \in \Pi}  \mathcal{P}\left(\pi, x \right) \leq \! \mathcal{P}\left(0, x\right) $ and 
$\displaystyle
- \frac{N_\zeta}{2}  \leq \!\sup_{x \in \mathcal{S}_2^+} \left[ \inf_{\pi \in \Pi}  \mathcal{P}\left(\pi, x \right) - \frac{1}{2} \zeta^*(2x)  \right] \leq \! \sup_{x \in \mathcal{S}_2^+} \left[  \mathcal{P}\left(0, x\right) - \frac{1}{2} \zeta^*(2x)  \right] \leq \frac{\widehat \zeta(1)}{2} $. 
\end{enumerate}
\end{proposition}
The proof, which is found in Appendix~\ref{app:PF}, is based on extensions of  techniques,   which are rather standard in spin glass theory.

\section{Main results} \label{sec:main}
\subsection{Quantum Parisi formula}

Our main result are three equivalent versions of a Parisi formula for the quantum glasses~\eqref{def:HN}. We start with the two versions  inspired and based on the corresponding formulae  for general vector-spin glasses established by Chen~\cite{Chen23}.
\begin{theorem}\label{thm:main}
Under Assumption~\ref{ass}:
\begin{align}\label{eq:thmmain}
\lim_{N\to \infty} \frac{1}{N} \mathbb{E}\left[ \ln Z_N \right] & = \sup_{x \in \mathcal{S}_2^+} \left[ \inf_{\pi \in \Pi}  \mathcal{P}\left(\pi, x \right) - \frac{1}{2} \zeta^*(2x)  \right]  \\
& =  \sup_{\varrho \in \mathcal{D} } \left[\inf_{x \in \mathcal{S}_2} \inf_{\pi \in \Pi}  \mathcal{P}\left(\pi, x \right) - \langle x , \varrho  \rangle+ \frac{1}{2} \zeta(\varrho)  \right]  .
\label{eq:TolandSinger}
\end{align}

\end{theorem}

As will be explained in the proof, which is found in Section~\ref{sec:PMain}, the second line~\eqref{eq:TolandSinger} is just an application of the Toland-Singer theorem to the supremum of the difference  of two convex functionals.  
In this aspect of our proof, we deviate from available results in~\cite{Chen23}.  In contrast, \cite{Chen23} establishes the equality of the right sides via the Hopf-Lax representation of the solution of an underlying Hamilton-Jacobi equation  (see also \cite{BG09,Mou21,ChXia22}). 

The outer supremum in~\eqref{eq:thmmain} is over non-negative, real hermitian Hilbert-Schmidt operators. As in~\cite{Chen23}, the role of $ x \in \mathcal{S}_2^+ $ is that of a Lagrange multiplier fixing the self-overlap in the functional integral $ Z_N $. 
For fixed $ x \in \mathcal{S}_2^+ $, the  inner infimum arises from a  Parisi formula for a self-overlap corrected functional integral spelled in Theorem~\ref{thm:mainsc} below. 
Both the infimum as well as the supremum are finite as seen from Proposition~\ref{prop:elpropP}. \\

We may also cast our Parisi formula in Theorem~\ref{thm:main} in a form analogous to Panchenko's version \cite{Pan18} of the Parisi formula for vector-spin glasses.
This version involves order parameters from the set of paths
\begin{equation}
 \Pi_\varrho \coloneqq \left\{ \pi: [0,1] \to \mathcal{D} \; \mbox{c\`adl\`ag} \  \big| \  0 \leq \pi(s) \leq \pi(t) \; \mbox{for all $ s \leq t $ with }\; \pi(0) = 0 \; \mbox{and} \; \pi(1) = \varrho \right\} ,
\end{equation}
terminating at  a fixed $ \varrho \in \mathcal{D} $. 
In the present context, the analogue of \cite[Thm. 1] {Pan18}  or \cite[Prop.3.1]{Chen23b} reads as follows.
\begin{theorem}\label{thm:Pan}
Under Assumption~\ref{ass}:
\begin{equation}\label{eq:mainPan}
\lim_{N\to \infty} \frac{1}{N} \mathbb{E}\left[ \ln Z_N \right]  =  \sup_{\varrho \in \mathcal{D} } \left[\inf_{x \in \mathcal{S}_2} \inf_{\pi \in \Pi_{\varrho}}  \mathcal{P}\left(\pi, x \right) - \langle x , \varrho  \rangle + \frac{1}{2} \zeta(\varrho)  \right]  .
\end{equation}
\end{theorem}
The proof can be found in Section~\ref{sec:PMain}. 
Comparing with~\eqref{eq:TolandSinger}, the optimal  $ \pi $ there is seen to end at $ \varrho $, which turns our to have the meaning of the self-overlap in the model~(cf.\  Subsection~\ref{sub:conj} below). Up to a limit and minor details in the variational sets, the functional and variational form in the right side of~\eqref{eq:mainPan} agrees with the one in~\cite{AdBr20}. 

We follow up with some further comments 
on special cases and implications of Theorems~\ref{thm:main} and~\ref{thm:Pan}.\\

\noindent
\textbf{Classical case $ b=0 $:}~~The version~\eqref{eq:mainPan} of the Parisi fromula is closest to the classical Parisi formula without transversal field. In case $ b = 0 $, where the Poisson measure $ \nu_1 $ has no jumps, $ \xi(t) = \xi(0) $ for all $ t \in [0,1) $.  The weight $ \langle \xi , x \xi \rangle = \langle 1 , x 1 \rangle = \int_0^1 \int_0^1 \widehat x(s,t) ds dt $ in the Parisi functional is hence a real number, $ \langle \xi , d\zeta(\pi(1)) \xi \rangle = \langle 1 ,  d\zeta(\varrho) 1 \rangle  $ and $ \langle w_\alpha , \xi \rangle =  \langle w_\alpha , 1 \rangle \xi(0) $. Consequently, the infimum in~\eqref{eq:mainPan} over $ x \in \mathcal{S}_2 $ forces $ \varrho $ to coincide with the rank-one projection $| 1 \rangle\langle 1 | $ onto the constant function $ 1 \in L^2(0,1) $, because otherwise the infimum is minus infinity. Since 
$ \zeta( | 1 \rangle\langle 1 |) - \langle 1 , d\zeta(| 1 \rangle\langle 1 |) 1 \rangle = \theta( | 1 \rangle\langle 1 | ) = \widehat\zeta'(1) - \widehat\zeta(1) \eqqcolon \widehat \theta(1)$,  
the  right side in~\eqref{eq:mainPan} therefore reduces at $ b = 0 $ to
$$
 \inf_{\pi \in \Pi_{|1\rangle\langle 1 |} } \mathbb{E}\Big[ \ln \sum_{\alpha \in \mathcal{N}^r}  \frac{\mu_\alpha^\pi}{2 \cosh(\beta h) } \sum_{\zeta(0) \in \{-1,1\}} \exp\left[  \left( \langle w_\alpha , 1 \rangle - \beta h\right)  \xi(0) \right] \Big]  + \frac{1}{2} \int_0^1 \theta(\pi(m)) dm - \frac{1}{2}\widehat \theta(1)  ,
$$
with $ \mathbb{E}\left[  \langle w_\alpha , 1 \rangle  \langle w_{\alpha'} , 1 \rangle \right] = \langle 1 , d\zeta(\pi(m_{\alpha\wedge\alpha'})) 1 \rangle $. To any given $ \pi \in \Pi_{|1\rangle\langle 1 | }$,  we may hence associate a classical order parameter $ \pi_c: [0,1] \to [0,1] $ ending at $  \pi_c(1) = 1 $ by setting 
$$
 \widehat\zeta'\left( \pi_c(m)\right) \coloneqq  \langle 1 , d\zeta(\pi(m))) 1 \rangle .
 $$
Convexity of $ \zeta $ implies $ \theta(\varrho_1) \geq \langle \varrho_2 , d\zeta(\varrho_1) \rangle - \zeta(\varrho_2) $ for any $ \varrho_1,  \varrho_2 \in \mathcal{D} $. The choice $ \varrho_1= \pi(m) $ and $ \varrho_2 = \pi_c(m) |1\rangle\langle 1 | $ leads to the lower bound for any $ m \in [0,1] $:
\begin{equation}\label{eq:lowerthetaconx}
\theta(\pi(m)) \geq \pi_c(m)  \langle 1 , d\zeta(\pi(m))) 1 \rangle - \widehat\zeta( \pi_c(m)) =  \pi_c(m)  \widehat\zeta'\left( \pi_c(m)\right) -  \widehat\zeta( \pi_c(m)) \eqqcolon \widehat \theta(  \pi_c(m)) .
\end{equation}
Therefore, the above infimum over quantum order parameters $ \pi \in \Pi_{|1\rangle\langle 1 |} $ can be reduced to an infimum over classical order parameters,  $$ \Pi_{c} \coloneqq \left\{ \pi : [0,1] \to [0,1] \; \mbox{c\`adl\`ag} \  \big| \  0 = \pi(0) \leq \pi(s) \leq \pi(t) \leq \pi(1) = 1  \; \mbox{for all}\;  s \leq t  \right\}  , $$
and one arrives at the familiar form \cite{Mezard:1986aa,Tal11b,Pan13} of the Parisi variational expression in the classical case:
\begin{equation}\label{eq:classical}
 \inf_{\pi_c \in \Pi_{c} } \mathbb{E}\Big[ \ln \sum_{\alpha \in \mathbb{N}^r}  \frac{\mu_\alpha^{\pi_c}}{2 \cosh(\beta h) } \sum_{\zeta(0) \in \{-1,1\}} \exp\left[  \left( \widehat w_\alpha - \beta h\right)  \xi(0) \right] \Big]  + \frac{1}{2} \int_0^1 \widehat\theta(\pi_c(m)) dm - \frac{1}{2}\widehat \theta(1)  
\end{equation}
with centred Gaussian random variables $ \widehat w_\alpha $ with covariance $  \mathbb{E}\left[  \widehat w_\alpha \widehat w_{\alpha'} \right] = \widehat\zeta'(\pi_c(m_{\alpha\wedge\alpha'})) $.\\

\noindent
\textbf{Annealed bound:}~~The choice $ \pi = 0 $ in~\eqref{eq:thmmain} results in the annealed functional, 
$$
\mathcal{P}\left(0, x \right) = \ln   \int \exp\left( \langle \xi , x  \xi  \rangle \right) \nu_1(d\xi) .
$$
As a corollary of Theorem~\ref{thm:main}, it enters an upper bound on the true quantum free energy:
$$
\limsup_{N\to \infty} \frac{1}{N} \mathbb{E}\left[ \ln Z_N \right]  \leq \sup_{x \in \mathcal{S}_2^+} \left[\mathcal{P}\left(0, x \right) - \frac{1}{2} \zeta^*(2x)  \right]   
=   \sup_{\varrho \in \mathcal{D} } \left[\inf_{x \in \mathcal{S}_2}  \mathcal{P}\left(0, x \right) - \langle x , \varrho  \rangle+ \frac{1}{2} \zeta(\varrho)  \right]  .
$$

\bigskip
\noindent
\textbf{QSK:}~~For the quantum Sherrington-Kirkpatrick model, for which
$$ \zeta(\varrho) = \frac{\beta^2}{2} \| \varrho \|_2^2 , \quad\mbox{and}\quad \zeta^*(x) = \frac{1}{2\beta^2} \| x \|_2^2  ,
$$
the free energy in the high-temperature regime $ \beta < 1 $ was shown~\cite{Leschke:2021aa} to be described by the annealed functional,
\begin{equation}\label{eq:QSKannealed} 
\lim_{N\to \infty} \frac{1}{N} \mathbb{E}\left[ \ln Z_N \right]  =  \sup_{x \in \mathcal{S}_2^+} \left[\mathcal{P}\left(0, x \right) - \frac{1}{\beta^2} \left\| x \right\|^2_2  \right] , \qquad \mbox{for all $ \beta < 1 $, $ b \in \mathbb{R} $ at $ h = 0 $.}
\end{equation}
The result~\cite{Leschke:2021aa} extends the high-temperature result~\cite{ALR87} for $ b = 0 $.  Extending another classical argument from~\cite{ALR87}, its was shown in~\cite{Leschke:2021ab}   that the spin-glass phase with a non-zero value of the Edwards-Anderson parameter
\begin{equation}\label{eq:EA}
	q\coloneqq \liminf_{N\to \infty}   \mathbb{E}\left[ \left( \langle r_N(\bxi^{(1)},  \bxi^{(2)} ) \rangle^{\otimes2}  - \mathbb{E}\left[ \langle r_N(\bxi^{(1)},  \bxi^{(2)} ) \rangle^{\otimes2} \right] \right)^2\right] 
\end{equation}
persists in the low temperature regime $ \beta > 1 $ for sufficiently small $ b > 0 $.  Here $ \langle (\cdot) \rangle^{\otimes2}  $ stands for  two independent copies of the Gibbs measure corresponding to $ Z_N $. 
The argument in~\cite{Leschke:2021ab} was extended in~\cite{Itoi:2023aa}, where the universality of the Edwards-Anderson parameter with regard to the distribution was shown. The latter expands the universality result in~\cite{Crawford:2007aa}, which guarantees that the quantum free energy does only dependent on the first two moments of the distribution of the random interactions.  

Based on numerical analysis and approximations~\cite{FS86,Yamamoto:1987aa,Usadel:1987aa,Young:2017aa} (see also \cite{SIC13}), for the QSK it is conjectured that   there is no spin glass phase, i.e.\ $ q = 0 $, for all sufficiently large $ b > 0 $ even at zero temperature $ \beta = \infty $.  However, as is shown in~\cite{Leschke:2021aa}, in this regime the free energy ceases to be given by~\eqref{eq:QSKannealed}. Rather, we expect that it will be given by $ \mathcal{P}\left(0, x \right) - \frac{1}{\beta^2} \left\| x \right\|^2_2 $ with the dual variable $ x $ to the self-overlap corresponding to the parameters $ \beta,  b $. The self-overlap  is conjectured to be self-averaging (cf.\ Subsection~\ref{sub:conj} below).

\subsection{Self-overlap corrected partition function}

Behind Chen's results~\cite{Chen23} as well as our main result, Theorem~\ref{thm:main}, is the proof of a Parisi formula for a self-overlap corrected partition function of the form
\begin{equation}\label{eq:defWN}
W_N(x) \coloneqq  \int \exp\left( -\int_0^1 U\left(\bxi(s)\right) ds -  \frac{N}{2} \zeta\left(  r_N(\bxi, \bxi)  \right)  \right) e^{N \langle r_N(\bxi, \bxi) , x \rangle } \nu_N(d\bxi) ,
\end{equation}
where $ x \in \mathcal{S}_2 $ is arbitrary. The first exponential is an exponential martingale. The second exponential involving $ N \langle r_N(\bxi, \bxi) , x \rangle = \sum_{j=1}^N \langle \xi_j , x \xi_j \rangle  $, may be incorporated in the apriori measure of the path integral. 
We extend Chen's results~\cite{Chen23} is the following way.
\begin{theorem}\label{thm:mainsc}
For any $ x \in \mathcal{S}_2 $:
\begin{equation}\label{eq:mainsc}
\lim_{N\to \infty} \frac{1}{N} \mathbb{E}\left[ \ln W_N(x)  \right]  = \inf_{\pi \in \Pi}  \mathcal{P}\left(\pi, x \right)  \eqqcolon \mathcal{F}(x) .
\end{equation}
\end{theorem}

This theorems harbours a variety of useful consequences for the functional $ \mathcal{F} :  \mathcal{S}_2 \to \mathbb{R}$ defined by the right side. Among it is its convexity and Gateaux differentiability. 
Note that the prelimit 
$$ G_N(x) \coloneqq {N}^{-1} \mathbb{E}\left[ \ln W_N(x)  \right]  $$
is evidently convex and differentiable at every $ x \in \mathcal{S}_2 $. Its derivative $ dG_N(x) $ is the averaged self-overlap under the Gibbs expectation $ \langle \cdot \rangle_{x,N} $ corresponding to $ W_N(x) $ in~\eqref{eq:defWN},
$$
\langle dG_N(x),  y \rangle \coloneqq \mathbb{E}\left[  \langle \langle r_N( \bxi ,  \bxi ),  y \rangle \rangle_{x,N}  \right] , \qquad \mbox{for any $ y \in \mathcal{S}_2 $.}
$$
Clearly, the right side defines a non-negative, rank-$ N $ operator   $ dG_N(x) \in \mathcal{S}_2^+ $.
Identical arguments as in the case of vector-spin models \cite[Prop.~2.3.]{Chen23b} allow to establish the differentiability of $ \mathcal{F} $. On top of that, we identify properties of its Legendre transform, which will be needed for establishing and exploiting the Toland-Singer identity~\eqref{eq:TolandSinger}. 

\begin{corollary}\label{cor:differble}
The functional $ \mathcal{F}: \mathcal{S}_2 \to \mathbb{R} $, $ \displaystyle \mathcal{F}(x) \coloneqq  \inf_{\pi \in \Pi}   \mathcal{P}\left(\pi, x \right)  $ is $ 1 $-Lipschitz continuous and convex with domain $ \dom \mathcal{F} = \mathcal{S}_2 $. 
\begin{enumerate}
\item Its Legendre transform $ \mathcal{F}^*: \mathcal{S}_2 \to [0,\infty)  $,
\begin{equation}\label{eq:defLegF}
 \mathcal{F}^*(\varrho) \coloneqq \sup_{x \in \mathcal{S}_2} \left( \langle x, \varrho \rangle - \mathcal{F}(x) \right) 
\end{equation}
is a proper, lower semicontinuous and convex functional with $ \dom \mathcal{F}^* \subset \mathcal{D}  $. %the latter being the closure of the convex set $ \mathcal{D} $. 
Moreover, $ \mathcal{F} = \left( \mathcal{F}^*\right)^* $. 
\item
$ \mathcal{F} $ is Gateaux differentiable at every $ x \in \mathcal{S}_2 $ with derivative $ d\mathcal{F}(x) $ given by 
\begin{equation}\label{eq:fvconv}
\lim_{N\to \infty} \langle dG_N(x),  y \rangle  = \langle d\mathcal{F}(x) , y \rangle  , \qquad \mbox{for any $ y \in \mathcal{S}_2 $.}
\end{equation}
Under the averaged Gibbs measure $ \langle \cdot \rangle_{x,N} $, the self-overlap weakly concentrates in the sense that for all $ y \in \mathcal{S}_2 $:
\begin{equation}\label{eq:concentration}
\lim_{N\to \infty} \mathbb{E}\left[\left|  \langle   \langle r_N( \bxi ,  \bxi ),  y \rangle \rangle_{x,N}  - \mathbb{E}\left[  \langle \langle r_N( \bxi ,  \bxi ),  y \rangle \rangle_{x.N}\right]  \right| \right] = 0 .
\end{equation}
\end{enumerate}
\end{corollary} 
Although a large part of the proof of 2.\ closely follows the arguments in~\cite{Chen23b}, we spell  it in  Subsection~\ref{sec:Cor1.9} with an eye on the pitfalls of the infinite-dimensional situation in the quantum set-up. Part of the argument relies on the self-averaging property of $ \ln W_N(x) $. Since this property is fundamental to the free energy of quantum spin glasses, we recall it from~\cite{Crawford:2007aa}. Note that in the present case of Gaussian processes $ U $, which are linear combinations of $ p $-spin interactions~\eqref{eq:pspinU}, the following proposition is a straightforward consequence of Gaussian concentration for Lipschitz functions.  
\begin{proposition}[cf.~\cite{Crawford:2007aa}] \label{Prop:self-averaging}
There are constants $ c , C \in (0,\infty) $ such that for any of the mixed-$p$-spin models described in Assumption~\ref{ass}, for any $ N \in \mathbb{N} $, $ x \in \mathcal{S}_2 $ and  $ t > 0 $:
\begin{align*}
\mathbb{P}\left( \left|  \ln W_N(x)  - \mathbb{E}\left[ \ln W_N(x) \right] \right| \geq \sqrt{N \ \widehat \zeta(1)} t  \right)  \leq C \exp\left( - c t^2 \right) ,
\end{align*}
and $  \mathbb{P}\left( \left|  \ln Z_N - \mathbb{E}\left[ \ln Z_N\right] \right| \geq \sqrt{N \ \widehat \zeta(1)} t  \right)   \leq C \exp\left( - c t^2 \right)$. 
\end{proposition}

\subsection{Structure of the order parameter and conjectures}\label{sub:conj}

Similarly as  in~\eqref{eq:concentration} for the self-overlap corrected Gibbs measure, it is reasonable to believe that the self-overlap $r_N(\bxi ,  \bxi ) $ asymptotically as $ N \to \infty $ concentrates with respect to the Gibbs measure corresponding to $ Z_N $. In the classical case $ b= 0 $, this is evident since  the self-overlap is identically one, $\widehat  r_N(\bxi ,  \bxi )(s,t) = 1 $ for all $ s,t \in [0,1) $.  If this conjecture is true also in the quantum case, then the self-overlap coincides with its Gibbs average $ \langle r_N(\bxi ,  \bxi ) \rangle $, which is clearly invariant under time shifts in the path integral. More precisely, the integral kernel of the self-overlap, $\widehat  r_N(\bxi ,  \bxi )(s,t)  $, would coincide with  its 
 time-averaged version. For a general Hilbert-Schmidt operator with kernel $ \widehat x $, the time-averaged Hilbert-Schmidt operator $ \overline{x} $ is defined through its  kernel
\begin{equation}\label{def:timeaverage}
\widehat {\overline{x}}(s,t) \coloneqq \int_0^1 \widehat x(s+\tau,t+\tau) \ d\tau  .
\end{equation}
Here and in the following, addition is understood on the unit torus, i.e. modulo one. 
Any such operator is diagonal in the orthonormal basis of $ L^2(0,1) $ composed of sine and cosine functions, of which the constant $ 1 \in L^2(0,1) $ is a member. 

The self-averaging of the self-overlap would have far reaching consequences. The outer supremum in~\eqref{thm:Pan} could be omitted by setting $ \varrho = \overline{\varrho} =  \lim_{N\to \infty} \mathbb{E}\left[ \langle r_N(\bxi ,  \bxi ) \rangle \right] $. The infimum over $ x $ in~\eqref{eq:thmmain} and~\eqref{thm:Pan} could also be omitted by setting $ x \in \mathcal{S}_2^+ $ the dual parameter, i.e. $ x = \overline{x} $ with $ \langle \overline{x} , \overline{\varrho} \rangle = \frac{1}{2} \left( \zeta(\overline{ \varrho}) + \zeta^*(2\overline{x} )\right) $. 

We furthermore conjecture that the remaining infimum over the Parisi order parameter can then be restricted. For the QSK, our final result supports this conjecture. In this case, the infimum can be restricted  to the two-parameter family which is spanned by $  \overline{\varrho} $ and the rank-one projection onto the constant function:
  \begin{align*}
 \Pi_{\overline{\varrho}}^c  \coloneqq \left\{ \pi \in  \Pi_{\overline{\varrho}} \ \big| \ \mbox{there is $  \kappa, \lambda: [0,1] \to [0,1] $ such that}\ \pi(m) =   \kappa(m)  \overline{\varrho} + \lambda(m) | 1 \rangle \langle 1 | \; 
  \right.   \qquad & \notag \\  \mbox{with $ \kappa(0) = \lambda(0) = \lambda(1) = 0 $ and $  \kappa(1) = 1 $}   &  \left.  \right\} . 
\end{align*}
This is  along the lines of simplifications found in early replica-calculations~\cite{FS86,Yamamoto:1987aa,Usadel:1987aa,BSS96}. 
\begin{theorem} \label{thm:QSK}
For the quantum Sherrington-Kirkpatrick model, for which $ \zeta(\varrho) = \frac{\beta^2}{2} \| \varrho \|_2^2 $:
\begin{equation}\label{eq:submanifold}
\inf_{\pi \in \Pi_{\overline{\varrho}}}  \mathcal{P}\left(\pi, \overline{x} \right)  = \inf_{\pi \in \Pi_{\overline{\varrho}}^c}  \mathcal{P}\left(\pi, \overline{x} \right) .
 \end{equation}
\end{theorem} 
Using the spectral decomposition, we may rewrite $ \overline{\varrho} = \overline{\varrho}^\perp + \langle 1 , \overline{\varrho}\ 1 \rangle \ | 1 \rangle \langle 1 | $  with the first term on the right collecting all eigenprojections orthogonal to the constant eigenfunction $ 1 \in L^2(0,1) $. The function $ \pi \in  \Pi_{\overline{\varrho}}^c  $ may thus be rewritten as
$$
\pi(m) =  \kappa(m) \  \overline{\varrho}^\perp + \widetilde\lambda(m) \  | 1 \rangle \langle 1 | 
$$
with $ \kappa, \widetilde\lambda : [0,1] \to [0,1] $ monotone non-decreasing with $ \kappa(0) = \widetilde\lambda(0) = 0 $ and $ \kappa(1) = 1 $ and $ \widetilde\lambda(1) =  \langle 1 , \overline{\varrho}\ 1 \rangle  $.  
Theorem~\ref{thm:QSK} is therefore similar in spirit to the analogous observation in the balanced Potts glass, where the Parisi order parameter also reflects the symmetry of the species~\cite{Elderfield:1983}. This was recently proven~\cite{BaSo23} under the assumption of the self-averaging of the self-overlap. Our short proof of Theorem~\ref{thm:QSK} is  simpler than the one in~\cite{BaSo23}. As it stands, it however also crucially relies on the fact that $ d\zeta(\varrho) = \beta^2 \varrho $ is linear in $ \varrho $. 

\begin{proof}[Proof of Theorem~\ref{thm:QSK}]
The inclusion  $ \Pi_{\overline{\varrho}}^c  \subset \Pi_{\overline{\varrho}} $, implies the '$ \leq $' part of~\eqref{eq:submanifold}. For the reverse inequality '$ \geq $', we 
fix $ \pi \in \Pi_{\overline{\varrho}} $ and set 
$$
\kappa(m) \coloneqq \sup\left\{ \kappa \in [0,1]  \ | \ \kappa \  \overline{\varrho} \leq  \pi(m) \right\}  .
$$
This quantity is well defined for any $  m \in [0,1] $, monotone-nondecreasing in $ m $, and we have $ \kappa(0) = 0 $ since $ \pi(0) = 0 $, and $ \kappa(1) = 1 $ since $ \pi(1) =  \overline{\varrho}  $.  We then define two independent $L^2 $-valued, centred Gaussian processes $ u_\alpha , v_\alpha $ indexed by the Ruelle probability cascade $ \alpha \in \mathbb{N}^r $ corresponding to $ \pi $:
\begin{align}
	\mathbb{E}\left[ u_\alpha u_{\alpha'} \right] = d\zeta\left( \pi(m_{\alpha \wedge \alpha'}) \right) - d\zeta\left(  \kappa( m_{\alpha \wedge \alpha'}) \overline{\varrho} \right) , \qquad \mathbb{E}\left[ v_\alpha v_{\alpha'} \right] = d\zeta\left(  \kappa( m_{\alpha \wedge \alpha'}) \overline{\varrho} \right) . 
\end{align}
Thanks to the monotonicity of $ d\zeta $ established in Lemma~\ref{lem:zetaprop}, the operator difference in the first equation is non-negative and hence a covariance.
The above processes give rise to the decomposition $ z_\alpha = v_\alpha + u_\alpha $ of the orginal $ L^2 $-valued, centred Gaussian process corresponding to $ \pi $ and defined in~\eqref{eq:defz}. 
Since the measure $ \nu_1 $ is invariant under the shifts $ \xi \mapsto \xi(\cdot + \tau) $ for any $ \tau \in (0,1) $, we may rewrite
\begin{align}\label{eq:X0av}
X_0(\pi,\overline{x}) & = \mathbb{E}\left[ \ln \sum_{\alpha \in \mathbb{N}^r } \mu_\alpha^\pi    \int \exp\left( \langle u_\alpha( \cdot - \tau) + v_\alpha( \cdot - \tau)  , \xi \rangle + \langle \xi , \big( \overline{x} - \tfrac{1}{2} d\zeta\left( \overline{\varrho}\right) \big)  \xi \rangle \right) \nu_1(d\xi)  \ \right]  \notag \\
& =   \mathbb{E}\left[ \ln \sum_{\alpha \in \mathbb{N}^r } \mu_\alpha^\pi    \int \exp\left( \langle u_\alpha( \cdot - \tau) , \xi \rangle + \langle v_\alpha  , \xi \rangle + \langle \xi , \big( \overline{x} - \tfrac{1}{2} d\zeta\left( \overline{\varrho}\right) \big)  \xi \rangle \right) \nu_1(d\xi)  \ \right]  .
\end{align}
Note that the translation invariance of $  \overline{x}  $ and $  \overline{\varrho} $ imply that the quadratic terms in $ \xi $ in the exponential is unaffected by the shift of $ \xi $. 
This also implies that the distribution of the process $ v_\alpha $ equals that of $  v_\alpha( \cdot - \tau)  $ for any $ \tau \in (0,1) $, which results in the second equality. 

We may now average $ X_0(\pi,\overline{x}) $ over $ \tau \in [0,1] $. Since the function $$ p \mapsto    \mathbb{E}\left[ \ln \sum_{\alpha \in \mathbb{N}^r } \mu_\alpha^\pi    \int \exp\left( p  \langle u_\alpha( \cdot - \tau) , \xi \rangle +(1-p) \langle u_\alpha( \cdot - \tau') , \xi \rangle+ \langle v_\alpha  , \xi \rangle + \langle \xi , \big( \overline{x} - \tfrac{1}{2} d\zeta\left( \overline{\varrho}\right) \big)  \xi \rangle \right) \nu_1(d\xi)  \ \right]  $$ is convex for $ \tau,\tau' \in [0,1] $ arbitrary, an application of Jensen's inequality to the $ \tau $-average yields the lower bound
\begin{align*}
X_0(\pi,\overline{x}) \geq \mathbb{E}\left[ \ln \sum_{\alpha \in \mathbb{N}^r } \mu_\alpha^\pi\int \exp\left(  \langle  \overline{u}_\alpha+ v_\alpha  , \xi \rangle + \langle \xi , \big( \overline{x} - \tfrac{1}{2} d\zeta\left( \overline{\varrho}\right) \big)  \xi \rangle \right) \nu_1(d\xi)  \right] , \quad   \overline{u}_\alpha \coloneqq \int_0^1  u_\alpha( \cdot - \tau)d\tau . 
\end{align*}
The $ L^2 $-valued, centred Gaussian process $  \overline{u}_\alpha $ has a covariance
$$
\mathbb{E}\left[  \overline{u}_\alpha(s)   \overline{u}_{\alpha'}(t) \right] = \langle 1 , d\zeta\left( \pi(m_{\alpha \wedge \alpha'}) \right) 1 \rangle - \langle 1 ,  d\zeta\left(  \kappa( m_{\alpha \wedge \alpha'}) \overline{\varrho}\right) 1 \rangle  ,
$$
which is independent of $ s , t \in (0,1) $. We may hence define
$$
\pi_c(m) \coloneqq \kappa(m) \overline{\varrho} + \lambda(m) | 1\rangle \langle 1 | , \quad \mbox{with} \;   \lambda(m) \coloneqq  \langle 1 , \pi(m)  1 \rangle - \kappa( m) \langle 1 ,  \overline{\varrho} 1 \rangle .
$$
Then $ \pi_c \in  \Pi_{\overline{\varrho}}^c  $. Moreover,  the random process $ w_\alpha \coloneqq \overline{u}_\alpha+ v_\alpha $ is Gaussian with $ \mathbb{E}\left[ w_\alpha w_{\alpha'} \right] = d\zeta\left( \pi_c(m_{\alpha \wedge \alpha'}) \right) $ by the linearity of $ d\zeta(\varrho) = \beta^2 \varrho $. 

Convexity of $ \theta $ implies 
$
\theta(\pi(m)) \geq \theta(\overline{\pi(m)}) 
$ for every $ m \in [0,1] $. By definition of $ \kappa(m) $, we have  $\Delta(m) \coloneqq \pi(m) - \kappa(m) \overline{\varrho} \geq 0 $, and hence 
$ \overline{\pi(m)} - \kappa(m) \overline{\varrho}  \geq  \langle 1 , \overline{\Delta(m)} 1 \rangle |1 \rangle \langle 1 | $. Since $  \langle 1 , \overline{\Delta(m)} 1 \rangle =  \langle 1 , \Delta(m) 1 \rangle =   \lambda(m)   $, the monotonicity of $ \theta $ yields
$$
\theta(\pi(m)) \geq \theta(\pi_c(m))
$$
for any $ m \in [0,1] $. Hence $   \mathcal{P}\left(\pi, \overline{x} \right) \geq   \mathcal{P}\left(\pi_c, \overline{x} \right) $, which completes the proof. 
\end{proof}

%\textbf{Something about differentiability in $ \beta_{2p} $ and a relation to the EA order parameter?}. 

\section{Approximation results}	

Our proofs of the Parisi formulae in Theorems~\ref{thm:main},~\ref{thm:Pan} and~\ref{thm:mainsc} will proceed through approximating the quantum partition functions $ Z_N $ and $ W_N(x) $ by partition functions of classical vector-spin glasses. A key observation, already contained in~\cite{AdBr20}, is that the functional integral concentrates on paths with a limited number of jumps. This enables the comparison of $ Z_N $ to the partition function of suitably defined vector-spin glasses.  In contrast to the Fourier basis used in~\cite{AdBr20}, we base the approximations of paths on square-wave pulses. This approximation has good monotonicity properties with regard to the correlation functionals. Moreover, it is closer to the philosophy used in the physics literature  by which the single-spin process in the functional integral is related via the Trotter formula to a strongly coupled Ising system~\cite{SIC13}.

In this section, we proceed to study the resulting approximation properties of the various main players in the Parisi formula: the partition function, the Parisi functional and its ingredient, the correlation functional.

\subsection{Paths space and functional integral}

In order to discretize paths, we use the a dyadic decomposition $   t_k := \frac{ k }{2^D} $ with $ k = 0, 1, \dots , 2^D $  of the unit interval $(0,1] $ into equidistant intervals $ (t_{k-1},t_k] $ of lengths $ 2^{-D } $ with $ D \in \mathbb{N} $ the approximation parameter.
The square-wave pulses
$$
e_k := 2^{D/2} 1_{(t_{k-1},t_k]} , \quad k = 1, \dots , 2^D ,
$$
form an orthonormal set in $ L^2(0,1) $. They span an increasing sequence of  nested, finite-dimensional subspaces 
$$
\mathcal{V}_D :=  \spa\left\{ e_k  \ | \  k = 1, \dots , 2^D \right\} \subset \mathcal{V}_{D+1} \subset L^2(0,1) . 
$$
Let $ Q_D \coloneqq \sum_{k=1}^{2^D} | e_k\rangle \langle e_k | $ stand for the orthogonal projection onto $ \mathcal{V}_D $. The projection of any $ f \in L^2(0,1) $ onto $ \mathcal{V}_N $ corresponds to the approximation of $f $ by
$$
f^D(s) := \left(Q_D f\right)(s) = \sum_{k=1}^D 1_{(t_{k-1},t_k]}(s) \ f(k) , \quad \mbox{with}\;  f(k) \coloneqq  2^{D} \int_{t_{k-1}}^{t_k}\! f(t) dt .
$$
This approximation has the following elementary properties:
\begin{enumerate}
\item 
If $ f \in B_\infty^2:= \left\{ f \in L^2(0,1)  \ | \ \| f \|_\infty \leq 1 \right\}  $, then $ f^D \in B_\infty^2 $ for any $ D $. This in particular applies to the paths $ \xi_j:[0,1] \to \{-1,1\} $. 
\item
For any $ f \in L^2(0,1) $:
\begin{equation}\label{eq:L2conv}
 	\lim_{D \to \infty } \| \left( 1- Q_D\right) f \|_2 = 0 .
\end{equation}
For continuous $ f $ this is evident by uniform continuity. The general case follows by the density of continuous functions in $ L^2(0,1) $.  By the boundedness of paths $ \xi_j:[0,1] \to \{-1,1\}  $ for any $ j $, this implies the convergence
\begin{equation}\label{eq:Lqconv}
\lim_{D \to \infty } \| \left( 1- Q_D\right) \xi_j \|_q = 0 
\end{equation}
in any $ L^q $-norm with  $ q \in [1,\infty) $. 
\end{enumerate}
For a path $ \xi_j:[0,1] \to \{-1,1\}  $, the approximation $  \xi_j^D $ is encoded in the averages 
\begin{equation}\label{eq:defsigmaD}
\sigma_{j}(k) := 2^{D} \int_{t_{k-1}}^{t_k}\! \xi_j(t) dt \in [-1,1] , \qquad \mbox{with $ k = 1, \dots , 2^D $.} 
\end{equation}
Under the discrete averaging~\eqref{eq:defsigmaD}, the probability measure $ \nu_1 $ on paths gets pushed forward to a probability measure $\nu_1^D $  on the Borel measure space over $  [-1,1]^{2^D} $. 
The next crucial lemma shows that the measure $ \nu_1 $ is in fact heavily concentrated on the projected paths.
\begin{lemma}\label{lem:basicappr}
For any $ \varepsilon > 0 $:
$$
\nu_1\left( \left\{ \| (1- Q_D) \xi \|_2 \geq \varepsilon \right\} \right) \leq  \exp\left( - 2^D \varepsilon^2  \ln\left(\frac{ 2^D \varepsilon^2  }{\beta b  e }\right) \right) . 
$$
\end{lemma}
\begin{proof}
An elementary computation shows
$$
\| (1- Q_D) \xi \|_2^2 =  \sum_{k=1}^D \left[  \int_{t_{k-1}}^{t_k}\! \xi(s)^2 ds  - 2^D \left( \int_{t_{k-1}}^{t_k}\! \xi(s) ds \right)^2 \right] .
$$
In case $ \xi: [0,1] \to \{-1,1\}  $ is constant on $  (t_{k-1},t_k]  $, the term in the square-brackets vanishes. Otherwise this term is bounded from above by $ 2^{-D} $. Therefore the sum is  estimated by
$$
\| (1- Q_D) \xi \|_2^2 \leq 2^{-D}  \eta_1(1,\omega_1)  .
$$
in terms of the total number $  \eta_1(1,\omega_1) $ of Poisson points on $[0,1) $. An exponential Chebychev-Markov estimate thus yields
\begin{align*}
\nu_1\left( \left\{ \| (1- Q_D) \xi \|_2 \geq \varepsilon \right\} \right) & \leq \inf_{\lambda > 0 } e^{-\lambda \varepsilon^2} \int \exp\left( \lambda 2^{-D} \eta_1(1,\omega_1) \right) \nu_1(d\xi)  \\
& \leq   \inf_{\lambda > 0 } e^{-\lambda \varepsilon^2}  \frac{\cosh( \beta |h|)}{\cosh(\beta \sqrt{b^2+h^2})} \ \sum_{k=0}^\infty \frac{(\beta b)^{2k} }{(2k)!}  \exp\left( \lambda 2^{-D} 2 k\right) \\
& \leq  \inf_{\lambda > 0 } e^{-\lambda \varepsilon^2}  \cosh\left(\beta b e^{\lambda 2^{-D}} \right) \leq  \exp\left( - 2^D \varepsilon^2  \ln\left(\frac{ 2^D \varepsilon^2  }{\beta b  e }\right) \right)  .
\end{align*}
This concludes the proof. 
\end{proof}

\subsection{Partition functions}

The main goal in this subsection is to show that the partition function $Z_N $ as well as its self-overlap corrected version $ W_N(x) $  is approximated by 
\begin{align*}
Z_N^D & \coloneqq \int  \exp\left( - \int_0^1 U\left(\bxi^D(s)\right) ds \right) \nu_N(d\bxi)  , \\
W_N^D(x) & \coloneqq    \int  \exp\left( - \int_0^1 U\left(\bxi^D(s)\right) ds  -  \frac{N}{2} \zeta\left(  r_N(\bxi^D, \bxi^D)  \right) + N  \langle  r_N(\bxi^D, \bxi^D)    , x\rangle\right)    \nu_{N} (d\bxi)  ,
\end{align*}
in which we projected the paths in the action functionals to the square-wave pulses $ \bxi^D \coloneqq (Q_D \xi_1, \dots, Q_D \xi_N) $. 
%We have introduced the abbreviation $ \nu_{N,x}  \coloneqq  \nu_{1,x} ^{\otimes N} $ for the $ N $-fold product of the modified apriori measure, which is a probability measure absolutely continuous with respect to $ \nu_1 $ with density 
%$$
%\frac{d\nu_{1,x} }{d\nu_1}(\xi) \coloneqq e^{\langle \xi , x \xi \rangle } . 
%$$

By construction, $ Z_N^D  $ coincides with the partition function of a vector-spin glass with apriori measure $ \nu_1^D  $, i.e.\ the push-forward measure of $ \nu_1 $, on the $ 2^D $-component vectors $  \left( \sigma_1(k) \right)_{k =1, \dots , 2^D}  \in   [-1,1]^{2^D} $, and $ W_N^D(x) $ is a self-overlap corrected analogue. 
The (classical) random energy on $ N $ such vector-spins is the Gaussian process with mean zero and covariance given by
\begin{multline}\label{eq:correlation}
\mathbb{E}\left[  \int_0^1 U\left(\bxi^{(1,D)}(s)\right) ds  \int_0^1 U\left(\bxi^{(2,D)}(t)\right) dt \right] \\
=  N \zeta \left(   r_N\big(\bxi^{(1,D)}, \bxi^{(2,D)}\big)  \right) =  \frac{N }{2^{2D} } \sum_{k, k' =1}^{2^D} \widehat\zeta\left( \frac{1}{N} \sum_{j=1}^N \sigma_j^{(1)}(k) \sigma_j^{(2)}(k') \right)  
\end{multline}
for two copies $ \boldsymbol{\sigma}^{(1)} , \boldsymbol{\sigma}^{(2)} $ of $ N $ vector-spins with $ 2^D $ components. 
Panchenko's proof~\cite{Pan18} and, in particular, Chen's simplifications~\cite{Chen23} on the Parisi formula for the free energy of mixed $ p$-spin models with vector spins hence become applicable.  Similarly as in~\cite{AdBr20}, our proof is based on this observation and the following approximation result. 
\begin{proposition}\label{prop:approx}
For any $ x \in \mathcal{S}_2 $:
\begin{align*}
& \lim_{D\to \infty } \limsup_{N\to \infty} \frac{1}{N} \left| \mathbb{E}\left[ \ln Z_N - \ln Z_N^D \right] \right| = 0 , \\
& \lim_{D\to \infty } \limsup_{N\to \infty} \frac{1}{N} \left| \mathbb{E}\left[ \ln W_N(x) - \ln W_N^D(x) \right] \right| = 0 .
\end{align*}
\end{proposition}
\begin{proof}
Corollary~\ref{cor:restr} and Proposition~\ref{lem:interpol1} below immediately yield the claim.
\end{proof}

To establish the ingredients of the proof, Corollary~\ref{cor:restr} and Proposition~\ref{lem:interpol1}, we largely follows the strategy in~\cite{AdBr20} with the Fourier bases exchanged by the square-wave pulses. 
We first restrict the path measure $ \nu_N $ to paths $ \bxi =(\xi_1,\dots, \xi_N) : [0,1] \to \{- 1, 1\}^N $, whose empirical self-overlap is well approximated by the square-wave pulses
\begin{equation}\label{eq:defSND}
 S_N^D(\varepsilon) := \left\{ \bxi  \ \Bigg| \ \frac{1}{N} \sum_{j=1}^N \langle \xi_j , (1-Q_D) \xi_j \rangle < \varepsilon \right\} . 
\end{equation}
The measure  $ \nu_{N} $ heavily concentrate on this set. 
\begin{lemma}\label{lem:concentration}
For any $ \varepsilon > 0 $ and $ D , N \in \mathbb{N}$ the complement of the set in \eqref{eq:defSND} has exponentially small measure:
\begin{equation}\label{eq:concentrationS}
\nu_{N}\left(  S_N^D(\varepsilon)^c \right) \leq \exp\left( - N \varphi_{D,x}^*(\varepsilon) \right)    ,
\end{equation}
with 
$$ \varphi_{D}^*(\varepsilon)  := \sup_{t\geq 0} \left( t \varepsilon -  \varphi_{D}(t) \right) \quad  \mbox{and} \quad  \varphi_{D}(t)  := \ln \int \exp\left( t  \langle \xi | (1-Q_D) \xi  \rangle\right) \nu_{1}(d\xi)  .
$$ 
For any $ \varepsilon , L > 0$ there is $D(\varepsilon,L) \in \mathbb{N} $ such that $   \varphi_{D}^*(\varepsilon)  \geq 2L $ for all $ D \geq D(\varepsilon,L) $. 
\end{lemma} 
\begin{proof}
The first assertion~\eqref{eq:concentrationS} is a standard concentration estimate based on the product nature of the measure. For an estimate of the Legendre transform $  \varphi_D^*(\varepsilon)  $, we use the convexity  estimate $   \varphi_D^*(\varepsilon) \geq t (\varepsilon - \varphi'_D(t) ) $, which is valid for any $ t \geq 0 $, at $ t_0 = 4L/\varepsilon $. By the dominated convergence theorem and~\eqref{eq:L2conv}, the derivative  
$$
\varphi'_D(t) = \int \exp\left( t  \langle \xi | (1-Q_D) \xi  \rangle\right)   \langle \xi | (1-Q_D) \xi  \rangle \nu_1(d\xi) \Big/ \int \exp\left( t  \langle \xi | (1-Q_D) \xi  \rangle\right) \nu_1(d\xi) 
$$
converges pointwise to zero as $ D \to \infty $. Consequently, for any $ \varepsilon > 0 $ and fixed $ t_0\geq  0 $ there is $ D(\varepsilon, t_0) \in \mathbb{N} $ such that $  \varphi'_D(t_0) < \varepsilon/2 $ for all $ D \geq D(\varepsilon, t_0) $. This concludes the proof.
\end{proof}

The restricted partition functions
\begin{align*}
Z_N(D,\varepsilon) & \coloneqq \int_{ S_N^D(\varepsilon) }  \exp\left( -  \int_0^1 U\left(\bxi(s)\right) ds \right) \nu_N(d\bxi) , \\
W_N(D,\varepsilon,x) & \coloneqq   \int_{ S_N^D(\varepsilon) }  \exp\left( -  \int_0^1 U\left(\bxi(s)\right) ds -  \frac{N}{2} \zeta\left(  r_N(\bxi, \bxi)  \right)  \right) e^{N  \langle  r_N(\bxi, \bxi)  , x\rangle }  \nu_{N}(d\bxi) 
\end{align*}
are evidently a lower bound on $ Z_N $ respectively $ W_N(x) $. The concentration, Lemma~\ref{lem:concentration}, implies that this bound is exponentially sharp under a quenched average. The same applies to the respective partition functions 
\begin{align*}
Z_N^D(D,\varepsilon) & \coloneqq  \int_{ S_N^D(\varepsilon) }  \exp\left( -  \int_0^1 U\left(\bxi^D(s)\right) ds \right) \nu_N(d\bxi) \leq Z_N^D  , \\
W_N^D(D,\varepsilon,x) & \coloneqq    \int_{ S_N^D(\varepsilon) }   \exp\left( - \int_0^1 U\left(\bxi^D(s)\right) ds  -  \frac{N}{2} \zeta\left(  r_N\big(\bxi^D, \bxi^D\big)  \right)  \right)  e^{N  \langle  r_N(\bxi^D, \bxi^D)  , x\rangle }   \nu_{N} (d\bxi) \leq W_N^D(x) ,
\end{align*}
for the approximands. 
\begin{corollary}\label{cor:restr}
For any $ \varepsilon > 0 $, there is some $ D(\varepsilon) \in \mathbb{N} $ such that for all $ D \geq D(\varepsilon) $ and all $ N \in \mathbb{N} $: 
\begin{equation}\label{eq:restr} 
\mathbb{E}\left[ \left| \ln Z_N(D,\varepsilon)  - \ln Z_N \right| \right] \leq \frac{e^{-4 N \widehat\zeta(1)} }{1 - e^{ - 8N \widehat\zeta(1)} } . 
\end{equation}
The same bound applies to $Z_N^D(D,\varepsilon)  $ and $ Z_N^D $. \\

\noindent
For any $ \varepsilon > 0 $ and $ x \in \mathcal{S}_2 $, there is some $ D(\varepsilon,x) \in \mathbb{N} $ such that for all $ D \geq D(\varepsilon,x) $ and all $ N \in \mathbb{N} $: 
\begin{equation}\label{eq:restr1} 
\mathbb{E}\left[ \left| \ln W_N(D,\varepsilon,x)  - \ln W_N(x) \right| \right] \leq \frac{e^{-3 N \widehat\zeta(1)} }{1 - e^{ - 8N \widehat\zeta(1)} } . 
\end{equation}
The same bound applies to $W_N^D(D,\varepsilon,x)  $ and $ W_N^D(x) $. 
\end{corollary}
\begin{proof}
The left side in~\eqref{eq:restr1} is upper bounded by 
\begin{align}\label{eq:ratiobound_conc}
\mathbb{E}\left[\ln \frac{Z_N}{Z_N(D,\varepsilon) }   \right] & \leq  \mathbb{E}\left[ Z_N(D,\varepsilon)^{-1}   \int_{ S_N^D(\varepsilon)^c }  \exp\left( - \int_0^1 U\left(\bxi(s)\right) ds \right) \nu_N(d\bxi) \right] \notag \\
& \leq \frac{1}{ \nu_N\left(  S_N^D(\varepsilon)\right) } \int_{ S_N^D(\varepsilon)^c }  \mathbb{E}\left[ \exp\left(\int \int_0^1 U\left(\widetilde\bxi(s)\right) ds \ \widehat\nu_N^{D,\varepsilon} (d\widetilde\bxi)  -  \int_0^1 U\left(\bxi(s)\right) ds \right) \right] \nu_N(d\bxi)  \notag \\
& \leq e^{4N\widehat\zeta(1) } \  \frac{\nu_N\left(  S_N^D(\varepsilon)^c\right) }{ \nu_N\left(  S_N^D(\varepsilon)\right) } ,
\end{align}
where  $ \widehat\nu_N^{D,\varepsilon} $ abbreviates the conditional probability measure $ \nu_N( (\cdot)\  1_{ S_N^D(\varepsilon) } ) / \nu_N\left(  S_N^D(\varepsilon)\right) $. In the above bound, the first inequality results from the elementary estimate $ \ln (1+ x) \leq x $, with $ x= (Z_N - Z_N(D,\varepsilon))/Z_N(D,\varepsilon) $. The second inequality is obtained from lower bounding $  Z_N(D,\varepsilon) $ using Jensen's inequality with respect to $ \nu_N $. 
Since the exponent in the second line is a Gaussian random variable with mean zero and covariance bounded by $ 4 N  \widehat\zeta(1) $, the third inequality follows.  Lemma~\ref{lem:concentration}, in which we choose $ L = 4 \widehat\zeta(1) $, thus concludes the proof. 

The same strategy yields the same upper bound for $Z_N^D(D,\varepsilon)  $ and $ Z_N^D $.

For a bound on $ \mathbb{E}\left[\ln \frac{W_N(x)}{W_N(D,\varepsilon,x) }   \right] $, we proceed as in~\eqref{eq:ratiobound_conc}. In the second line, the exponent then also contains the term
$$
\int  \frac{N}{2} \zeta\left( r_N\big(\widetilde\bxi,\widetilde\bxi\big) \right) \widehat\nu_{N}^{D,\varepsilon} (d\widetilde\bxi)  - \frac{N}{2} \zeta\left( r_N\big(\bxi,\bxi\big) \right)  \leq  \frac{N}{2}  \widehat\zeta(1) ,
$$
and
$$
 N \langle  r_N(\bxi^D, \bxi^D)  , x\rangle -\int N \langle  r_N(\bxi^D, \bxi^D)  , x\rangle \ \widehat\nu_{N}^{D,\varepsilon} (d\widetilde\bxi)   \leq 2 N \| x \|_2  \ \varepsilon  . 
$$
%with an appropriate probability measure  $\widehat\nu_{N}^{D,\varepsilon}  $, which results from the restriction of $ \widehat\nu_{N,x} $ to the set $ S_N(\varepsilon) $.  
The claim then again follows from Lemma~\ref{lem:concentration} with $ L = 4 \widehat\zeta(1) + 2 \| x \|_2 \varepsilon $. 

The bound on $  \mathbb{E}\left[\ln \frac{W_N^D(x)}{W_N^D(D,\varepsilon,x) }   \right] $ follows  similarly. 
\end{proof}

The following final input for the proof of Proposition~\ref{prop:approx} is a Gaussian interpolation bound.

\begin{proposition}\label{lem:interpol1}
There is some constant $ C \in (0,\infty ) $ such that for all $  \varepsilon  >0 $, $ D,N \in \mathbb{N} $ and $ x \in \mathcal{S}_2 $:
\begin{align*}
&  \left| \mathbb{E}\left[  \ln Z_N(D,\varepsilon)  - \ln  Z_N^D(D,\varepsilon)  \right]  \right| \leq C \sqrt{\varepsilon} \ N , \\
&  \left| \mathbb{E}\left[  \ln W_N(D,\varepsilon,x)  - \ln  W_N^D(D,\varepsilon,x)  \right]  \right| \leq C \sqrt{\varepsilon} \ N ,
\end{align*}
\end{proposition}
\begin{proof}
The proof is a standard Gaussian interpolation argument. For the first claim we control the derivative of 
\begin{equation}\label{eq:interpol1}
f(r) \coloneqq \frac{1}{N} \ \mathbb{E}\left[ \ln \int_{ S_N^D(\varepsilon) }  \exp\left( -  \sqrt{r} \int_0^1 U\left(\bxi(s)\right) ds - \sqrt{1-r} \int_0^1 U\left(\bxi^D(s)\right) ds \right) \nu_N(d\bxi)  \right] 
\end{equation}
Let $ \langle \cdot \rangle_r $ stand for the corresponding Gibbs expectation value of the interpolated action under the restricted path measure $  \widehat\nu_N^{D,\varepsilon} $, and $  \langle \cdot \rangle_r^{\otimes 2} $ for the Gibbs expectation value on a duplicated system. 
A straightforward differentiation followed by Gaussian integration by parts yields (cf.\ \cite[Ch. 1.2]{Pan13})
\begin{equation}\label{eq:der1}
f'(r) = \frac{1}{2} \ \mathbb{E}\left[  \langle\zeta\left( r_N(1,1) \right)\rangle_r  -  \langle\zeta\left( r_N^D(1,1)\right) \rangle_r \right] -  \frac{1}{2} \ \mathbb{E}\left[  \langle\zeta\left( r_N(1,2) \right)\rangle_r^{\otimes 2}  -  \langle\zeta\left( r_N^D(1,2)\right) \rangle_r^{\otimes 2} \right] ,
\end{equation}
where we abbreviated the Hilbert-Schmidt integral operators $ r_N(1,1) $, $ r_N^{D}(1,1) $ and  $ r_N(1,2) $, $ r_N^{D}(1,2) $, which are defined through their kernels
\begin{align*} 
\widehat r_N^{(D)}(1,1)(s,t) & :=  N^{-1} \bxi^{(1)(D)}(s) \cdot  \bxi^{(1)(D)}(t) \\
\widehat r_N^{(D)}(1,2)(s,t)  & := N^{-1} \bxi^{(1)(D)}(s) \cdot  {\bxi}^{(2)(D)}(t) ,
\end{align*} 
and $ \bxi^{(1)}$ respectively $\bxi^{(2)}$ refer to the paths in the first respectively second copy of the system.  Since $ \widehat\zeta:[0,1]\to \mathbb{R} $ was assumed to be Lipschitz continuous with constant $ L < \infty $, we conclude that
\begin{align*}
\left| \zeta\left( r_N(1,2)\right) - \zeta\left( r_N^D(1,2)\right) \right| & \leq L  \int_0^1\int_0^1  \left| N^{-1}  \bxi^{(1)}(s) \cdot   \bxi^{(2)}(t) - N^{-1} \bxi^{(1) D}(s) \cdot  \bxi^{(2) D}(t)\right| ds dt \\
&  \leq L \left\| r_N(1,2)  - Q_D r_N(1,2) Q_D \right\|_2 \leq 2 L  \left\| (1-Q_D) r_N(1,2)  \right\|_2  ,
\end{align*} 
where the last step is the triangle inequality for the Hilbert-Schmidt norm. 
These quantities are bounded using
\begin{align*}
\left\| (1-Q_D) r_N(1,2)  \right\|_2^2  & \leq    \frac{1}{N^2} \sum_{j,k=1}^N \left| \langle \xi_j^{(1)} | (1-Q_D) \xi_k^{(2)} \rangle \right| \\
&  \leq \sqrt{ \frac{1}{N} \sum_{j=1}^N \langle \xi_j^{(1)} | (1-Q_D) \xi_j^{(1)} \rangle   \  \frac{1}{N} \sum_{k=1}^N \langle \xi_k^{(2)} | (1-Q_D) \xi_k^{(2)} \rangle } \leq \varepsilon ,
\end{align*}
where the second inequality is a combination of the Cauchy-Schwarz inequality and Jensen's inequality, and the last inequality is valid due to the restriction of the support of the measures to the set in~\eqref{eq:defSND}.  
The difference $ \left|\zeta\left( r_N(1,1)\right) - \zeta\left( r_N^D(1,1)\right)  \right| $ is estimated similarly. Therefore, the derivative $ | f'(r) | $ is bounded by a constant, which is independent of $ D $, times $ \sqrt{\varepsilon} $. This completes the proof of the first assertion. 

For the second assertion we proceed analogously. Modifying  the exponentials in~\eqref{eq:interpol1} by the respective self-overlap correction, results in the cancellation of the first two terms in~\eqref{eq:der1} of the respective derivative. Proceeding analogously as above for the remaining terms yields the claim. 
\end{proof}

\subsection{Correlation functional} 
Our approximation argument will also rely on approximations of the quantum Parisi functional in terms of their projections onto the subspace of square-wave pulses. In this section, we will set the stage and consider the space
$$
\mathcal{D}_D := \left\{ \varrho \in  \mathcal{S}_{2,D}= \mathbb{R}^{2^D\times 2^D} \, \big| \, \varrho \geq 0, \; \left| \varrho_{k,k'} \right| \leq 1  \right\} 
$$
of non-negative $ 2^D \times 2^D $ matrices $ \mathcal{S}_{2,D}^+ $, whose elements are bounded by one. It will be identified with the projection of $ \mathcal{D} $ onto the subspace of square-wave pulses, i.e. $  Q_D \varrho \ Q_D $ for $ \varrho \in \mathcal{D}$.  This is accomplished by the embedding
\begin{align}\label{eq:QDrho}
I_D:  \ & \mathcal{S}_{2,D} \to \mathcal{S}_2 , 
\qquad I_D(\varrho)(s,s') = \sum_{k,k'=1}^{2^D} 1_{(t_{k-1}, t_k]}(s) 1_{(t_{k'-1}, t_{k'}]}(s')  \, \varrho_{k,k'},  %\notag 
\end{align} 
which has the following elementary properties, which follow by a computation directly from the definition:
\begin{enumerate}
\item $ I_D(\varrho) = Q_D  I_D(\varrho) Q_D $ for any $ \varrho \in  \mathcal{S}_{2,D} $
\item  Equipping $   \mathcal{S}_{2,D} $ with the scalar product
$$ \langle x, \varrho \rangle_D :=  \frac{1}{2^{2D}} \sum_{k,k'=1}^{2^D}  x_{k,k'} \varrho_{k',k} , $$
the embedding preserves the scalar product, i.e. $ \langle x, \varrho \rangle_D = \langle   I_D(x) ,  I_D(\varrho) \rangle $ for all $x, \varrho \in  \mathcal{S}_{2,D} $.
\item $\displaystyle \| \widehat{(I_D x)} \|_q^q = \frac{1}{2^{2D}} \sum_{k,k'=1}^{2^D}  |x_{k,k'}|^q $ for any $ x \in \mathcal{S}_{2,D} $ and $ q > 0 $.  
\end{enumerate}
We will also  use the inverse of $ I_D $ on its range, which coincides with the projection 
\begin{align}\label{def:JD} 
 J_D :  \ & \mathcal{S}_2 \to  \mathcal{S}_{2,D} , \qquad (J_D \varrho)_{k,k'} \coloneqq  2^{2D} \int_{t_{k-1}}^{t_k}  \int_{t_{k'-1}}^{t_{k'}}  \varrho(s,s')  \ ds ds'  , \quad k,k' \in \{ 1, \dots , 2^D\} , 
\end{align}
for which:
\begin{enumerate}
\item $ J_D \circ I_D = \mathrm{id}_{ \mathcal{S}_{2,D}} $, 
\item $ J_D( \varrho) = J_D (Q_D \varrho Q_D) $ for all $ \varrho \in  \mathcal{S}_{2} $. 
\end{enumerate}

\bigskip
The restriction of the functional $ \zeta : \mathcal{D} \to [0,\infty) $ to $ \mathcal{D}_D $ is
\begin{equation}\label{eq:defzetaD}
\zeta_D(\varrho) \coloneqq \frac{1}{2^{2D}} \sum_{k,k'=1}^{2^{2D}} \widehat\zeta( \varrho_{k,k'} ) .
\end{equation}
It trivially extends from $  \mathcal{D} $ to $  \mathcal{S}_{2,D}^+ $.  Its Legendre transform is defined on $ 2^D \times 2^D $ matrices $ x \in \mathcal{S}_{2,D}^+ $:
\begin{equation}\label{def:LenegdreD}
\zeta_D^*(x) \coloneqq \sup_{\varrho \in \mathcal{S}_{2,D}^+}\left( \langle x, \varrho \rangle_D -  \zeta_D(\varrho) \right)  .
\end{equation}
The following lemma summarises the relation of these  functionals to the correlation functional $ \zeta $ and its Legendre transform $ \zeta^* $.  
\begin{lemma}\label{lem:relproj}
\begin{enumerate}
\item For any  $ \varrho \in \mathcal{D} $ and $ D \in \mathbb{N} $: 
\begin{equation}\label{eq:monzetaD}
\zeta(Q_D \varrho Q_D) \leq \zeta(Q_{D+1} \varrho Q_{D+1}) \leq \zeta(\varrho) . 
\end{equation}
For any $ \varrho  \in \mathcal{D} $: \; $ \zeta( Q_D \varrho \ Q_D) =  \zeta_D(J_D(\varrho)) $.
\item For any $ x  \in \mathcal{S}_{2}^+ $ and any $ D \in \mathbb{N} $: \;$ \zeta^*(x) \geq   \zeta_D^*(J_D(x)) $ and:
\begin{equation}\label{eq:consiststar}
	\zeta^*\left( Q_D x Q_D\right) =  \zeta_D^*(J_D(x)) . 
\end{equation}
Moreover: $  \zeta^*\left( Q_{D+1} x Q_{D+1} \right)  \geq\zeta^*\left( Q_D x Q_D\right)  $.
  \item For any $ \varrho = Q_D \varrho  Q_D \in \mathcal{D} $ and $ x \in \mathcal{S}_2 $: \; $ \langle d\zeta(\varrho) , x \rangle = \langle d\zeta(\varrho) , Q_D x Q_D \rangle $.\\
  For any $ \varrho_D \in  \mathcal{D}_{D} $ and $  x_D \in  \mathcal{S}_{2,D} $: \; 
$$  \langle d\zeta_D(\varrho_D)  , x_D \rangle_D= \langle d\zeta\left( I_D(\varrho_D) \right) , I_D(x_D) \rangle  .
$$
\end{enumerate}
\end{lemma}
\begin{proof}
1.~The monotonicity~\eqref{eq:monzetaD} follows from the convexity of $ \zeta $ and the fact that $Q_D \varrho \ Q_D $ results from $ Q_{D+1} \varrho \ Q_{D+1} $ by averaging. The second claim is by straightforward computation using~\eqref{eq:QDrho}. \\
2.~For any $x \in \mathcal{S}_2 $ and $ \varrho = Q_D \varrho \ Q_D $, we have $ \langle x, \varrho \rangle  =  \langle Q_D x Q_D,  Q_D \varrho Q_D\rangle = \langle J_D(x) , J_D(\varrho)\rangle_D $. Therefore
\begin{align}\label{eq:supestforstar}
 \zeta^*(x)  & \geq \sup_{\substack{ \varrho \in \mathcal{S}_2^+ \\ \varrho = Q_{D} \varrho Q_{D} }} \left( \langle x, \varrho \rangle -  \zeta(\varrho) \right)  = \sup_{\substack{ \varrho \in \mathcal{S}_2^+ \\ \varrho = Q_D \varrho \ Q_D }}   \left(\langle J_D( x) , J_D(\varrho )\rangle_D -  \zeta(\varrho) \right) \\
 & =   \sup_{\varrho_D \in \mathcal{S}_{2,D}^+}\left( \langle J_D (x), \varrho_D \rangle_D -  \zeta_D(\varrho_D) \right)  = \zeta_D^*(J_D(x)) , \notag
\end{align} 
where we used the first part of the lemma in the last line. This finishes the proof of the first assertion. The second assertion follows from $ \zeta(\varrho) \geq \zeta(Q_D\varrho \ Q_D) $ established in 1., which leads to
$$
 \zeta_D^*\left(J_D( x)\right) = \sup_{\varrho \in \mathcal{S}_2^+} \left[ \langle Q_D \varrho\  Q_D  , x \rangle - \zeta(Q_D \varrho Q_D ) \right]  \geq 
 \sup_{\varrho \in \mathcal{S}_2^+} \left[ \langle \varrho , Q_D x Q_D \rangle - \zeta(\varrho) \right]  =
\zeta^*\left(Q_D x Q_D\right)  ,
$$
which, together with the first part,  concludes the proof of~\eqref{eq:consiststar}. The  last assertion follows from the above using $ Q_D Q_{D+1} = Q_D $. 

3.~The operator $ d\zeta(\varrho) $ representing the Gateaux derivative satisfies $  d\zeta(\varrho)  = Q_D d\zeta(\varrho) Q_D $ for $ \varrho = Q_D \varrho Q_D $. 
Therefore
$$
 \langle  d\zeta(\varrho), x\rangle =   \langle Q_D d\zeta(\varrho) Q_D , x\rangle =   \langle d\zeta(\varrho),  Q_D xQ_D \rangle . 
$$
The Gateaux derivative of $ \zeta_D $ is given by  $ d \zeta_D(\varrho_D) = \sum_{p\geq 1} \beta_{2p}^2 \varrho_D^{\odot 2p-1} $  in terms of the (ordinary) Hadamard products of the matrix $ \varrho_D \in \mathcal{D}_D $. 
Consequently,
$$
 \langle d\zeta_D(\varrho_D)  , x_D \rangle_D=  \sum_{p\geq 1} \beta_{2p}^2 \langle \varrho_D^{\odot 2p-1}, x_D \rangle_D =  \sum_{p\geq 1} \beta_{2p}^2 \langle I_D\big(\varrho_D^{\odot 2p-1}\big), I_D(x_D) \rangle = \langle d\zeta\left( I_D(\varrho_D) \right) , I_D(x_D) \rangle  ,
$$
where the last step used $ I_D\big(\varrho_D^{\odot 2p-1}\big) = I_D\big(\varrho_D\big)^{\odot 2p-1} $. 
\end{proof} 

The above lemma in particular guarantees that the functional $ \theta_D: \mathcal{S}_{2,D}^+ \to [0,\infty) $ defines by 
\begin{equation}\label{eq:defthetaD}
\theta_D( \varrho_D ) \coloneqq \langle d\zeta_D(\varrho_D) , \varrho_D \rangle_D - \zeta_D( \varrho_D) =  \theta\left( I_D( \varrho_D) \right) 
\end{equation}
agrees with the restriction of $ \theta $ to the projected values.

\subsection{Quantum Parisi functional}\label{sec:QPF}

A crucial step in the proof of our main results is the approximation of the quantum Parisi functional in terms of the Parisi functional of the vector-spin model, which corresponds to the projection onto square-wave pulses. The Lipschitz continuity, Proposition~\ref{lem:LcontP}, is a basic tool for this. It immediately implies the following lemma, which will be used in the proofs of Theorems~\ref{thm:Pan} and ~\ref{thm:mainsc}. 
\begin{lemma}\label{lem:infconv}
For any $ x \in \mathcal{S}_2 $:
\begin{equation}\label{eq:inf1}
\lim_{D\to \infty} \inf_{\pi \in \Pi}  \mathcal{P}(Q_D\pi Q_D , Q_D x Q_D  )  =  \inf_{\pi \in \Pi}  \mathcal{P}(\pi , x ) .  
\end{equation}
For any $ \varrho \in \mathcal{D} $:
\begin{equation}\label{eq:inf2}
\liminf_{D\to \infty}\inf_{x \in \mathcal{S}_2}  \inf_{\pi \in \Pi_\varrho}  \mathcal{P}(Q_D\pi Q_D , Q_D x Q_D  )  - \langle Q_D x Q_D , \varrho \rangle \geq  \inf_{x \in \mathcal{S}_2}  \inf_{\pi \in \Pi_\varrho}  \mathcal{P}(\pi , x ) - \langle x , \varrho \rangle . 
\end{equation}
\end{lemma}
\begin{proof}
By the Lipschitz continuity, Proposition~\ref{lem:LcontP},
$$
 \left| \inf_{\pi \in \Pi}  \mathcal{P}(Q_D\pi Q_D , Q_D x Q_D  )  -  \inf_{\pi \in \Pi}  \mathcal{P}(Q_D\pi Q_D , x   ) \right| \leq \| Q_D x Q_D - x  \|_2 .
$$ 
From~\eqref{eq:L2conv} we know that for every $ y \in \mathcal{S}_2 $:
\begin{equation}\label{eq:HSconv}
\lim_{D\to \infty} \left\| Q_{D} y Q_{D}  - y  \right\|_2 = 0 .
\end{equation}
Enlarging the set over which the infimum is taken by dropping the projections, we conclude
\begin{align*}
\inf_{\pi \in \Pi}  \mathcal{P}(Q_D\pi Q_D , x  ) \geq \inf_{\pi \in \Pi}  \mathcal{P}(\pi  , x ) .
\end{align*}
For a complementary upper bound, we recall from Proposition~\ref{prop:elpropP} that the infimum is finite.  For $ \varepsilon > 0 $ we may thus pick $ \pi_\varepsilon \in \Pi$ such that
\begin{align*}
 \inf_{\pi \in \Pi}  \mathcal{P}(\pi  , x ) + \varepsilon  \geq \mathcal{P}(\pi_\varepsilon  , x ) \geq  \mathcal{P}(Q_D \pi_\varepsilon Q_D  , x  ) - L  \int_0^1 \| Q_D\pi_\varepsilon(m)Q_D - \pi_\varepsilon(m) \|_2 dm ,
\end{align*}
where the last inequality is by Proposition~\ref{lem:LcontP}. 
The last integral converges to zero as $ D\to \infty $  by~\eqref{eq:HSconv}  and the dominated convergence theorem.
 Since $ \epsilon > 0 $ was arbitrary, this finishes the proof of~\eqref{eq:inf1}.  \\
 
 The proof of~\eqref{eq:inf2} proceeds analogously to the first step above. The Lipschitz continuity in $ x \in \mathcal{S}_2$, Proposition~\ref{lem:LcontP}, controls
 $$
  \left| \inf_{\pi \in \Pi_\varrho}  \mathcal{P}(Q_D\pi Q_D , Q_D x Q_D  ) -  \inf_{\pi \in \Pi_\varrho}  \mathcal{P}(Q_D\pi Q_D , x   ) +  \langle x - Q_D x Q_D , \varrho \rangle \right| \leq  2 \| Q_D x Q_D - x  \|_2 .
 $$
 Elarging the  set, over which the infimum over $ \pi \in \Pi_\varrho $ is taken, yields the estimate $ \inf_{\pi \in \Pi_\varrho}  \mathcal{P}(Q_D\pi Q_D , x  ) \geq \inf_{\pi \in \Pi_\varrho}  \mathcal{P}(\pi  , x ) $, which completes the proof of~\eqref{eq:inf2}. 
\end{proof}

Due to the additional presence of the supremum over $ x \in \mathcal{S}_2^+ $, the proof of Theorem~\ref{thm:main} will rely on a more subtle approximation analysis. To prepare this, we again restrict the path measure $ \nu_1 $ to the set
$$
S^D(\varepsilon) := \left\{ \| (1- Q_D) \xi \|_2 \leq \varepsilon \right\} .
$$
For $ \pi \in \Pi $ and $ x \in \mathcal{S}_2 $, we consider the restricted functionals
\begin{equation}\label{eq:PRuelle2}
 X_0(\pi,x;D,\varepsilon) \coloneqq  \mathbb{E}\left[ \ln \sum_{\alpha \in \mathbb{N}^{r} } \mu^\pi_\alpha \int_{S^D(\varepsilon)} \exp\left( \langle \xi ,  z_\alpha \rangle + \langle \xi , \left( x- \tfrac{1}{2} d\zeta(\pi(1) \right)  \xi  \rangle \right) \nu_1(d\xi) \right] ,
 \end{equation}
 as well as $ \mathcal{P}(\pi,x;D,\varepsilon) :=   X_0(\pi,x;D,\varepsilon)  + \frac{1}{2} \int_0^1 \theta(\pi(m)) dm $. \\
 
The analogue of Proposition~\ref{prop:approx} on the level of the quantum Parisi functional is the following. %
\begin{proposition}\label{prop:approxPF}
For any $ x \in \mathcal{S}_2 $ with integral kernel $ \widehat x $, any $ \varepsilon > 0 $ and all $ D \in \mathbb{N} $:
\begin{equation}\label{eq:PDapprx}
\sup_{\pi \in \Pi} \left|  \mathcal{P}(\pi,x)  - \mathcal{P}(\pi,x;D,\varepsilon) \right| \leq \exp\left( 4 N_\zeta + \| \widehat x \|_1 -  2^D \varepsilon^2  \ln\left(\frac{ 2^D \varepsilon^2  }{\beta b  e }\right) \right)  ,
\end{equation}
where we recall $ N_\zeta = \sum_{p\geq 1} 2p \beta_{2p}^2 $. 
\end{proposition} 
\begin{proof}
Since the proof proceeds analogously as that of Proposition~\ref{prop:approx}, we will only highlight the differences.
We abbreviate the random action
$$
\mathcal{A}(\alpha,\xi;\pi ,x) :=  \langle \xi ,  z_\alpha \rangle + \langle \xi , \left( x- \tfrac{1}{2} d\zeta(\pi(1) \right)  \xi  \rangle , 
$$
and estimate the left side in~\eqref{eq:PDapprx} for any $ \pi \in \Pi $ with $ r $ steps by
\begin{multline}\label{eq:Jensen}
\mathbb{E}\left[ \frac{\sum_{\alpha \in \mathbb{N}^{r} } \mu^\pi_\alpha \int_{S^D(\varepsilon)^c} \exp\left(\mathcal{A}(\alpha,\xi;\pi ,x)  \right) \nu_1(d\xi) }{ \sum_{\alpha \in \mathbb{N}^{r} } \mu^\pi_\alpha \int_{S^D(\varepsilon)} \exp\left(\mathcal{A}(\alpha,\xi;\pi ,x)  \right) \nu_1(d\xi) }  \right] \\
  \leq  \nu_1\left(S^D(\varepsilon)\right)^{-1} \sum_{\alpha \in \mathbb{N}^{r} } \mu^\pi_\alpha \int_{S^D(\varepsilon)^c} \mathbb{E}\left[\exp\left(\mathcal{A}(\alpha,\xi;\pi ,x) - \langle  \mathcal{A}(\cdot ;\pi ,x)\rangle_\varepsilon \right) \right]  \ \nu_1(d\xi) \qquad
\end{multline}
where we used the normalization $ \sum_{\alpha \in \mathbb{N}^{r} } \mu^\pi_\alpha = 1 $ and Jensen's inequality for the probability measure 
$$ 
\sum_{\alpha \in \mathbb{N}^{r} } \mu^\pi_\alpha \int_{S^D(\varepsilon)} (\cdot)  \ \overline{\nu}^\varepsilon_1(d\xi)  =: \langle (\cdot) \rangle_\varepsilon \quad \mbox{with}\quad  \overline{\nu}^\varepsilon_1(\cdot) :=  \nu_1\left(S^D(\varepsilon)\right)^{-1} \nu_1(\cdot) .
$$
The non-random terms in the exponential in the right side of~\eqref{eq:Jensen} involving $ x \in \mathcal{S}_2^+ $ and $ \pi(1) $ are bounded from above by $ 0 \leq \langle \xi , x \xi \rangle \leq \| \widehat x \|_1 $ and $0 \leq \langle \langle \xi , d\xi(\pi(1))  \xi \rangle \rangle_\varepsilon \leq N_\zeta  $ by Lemma~\ref{lem:zetaprop}.  The random terms  in the exponential in the right side of~\eqref{eq:Jensen} is the Gaussian random variable $ \langle \xi ,  z_\alpha \rangle  - \langle\langle \cdot,  z_{\cdot} \rangle   \rangle_\varepsilon$. The latter is centered and has a covariance, which is bounded as follows
\begin{align*}
& \mathbb{E}\left[ \left| \langle \xi ,  z_\alpha \rangle  - \langle \langle \xi ,  z_\alpha \rangle  \rangle_\varepsilon \right|^2\right]   \leq 2 \ \mathbb{E}\left[ \langle \xi ,  z_\alpha \rangle^2 \right] + 2 \ \mathbb{E}\left[ \langle\langle \cdot,  z_{\cdot} \rangle   \rangle_\varepsilon^2 \right]  \\
& \leq  2 \ \langle \xi , d\zeta(\pi(1) ) \xi \rangle + 2  \sum_{\alpha,\alpha' \in \mathbb{N}^{r} } \mu^\pi_\alpha  \mu^\pi_{\alpha'} \int_{S^D(\varepsilon)}  \int_{S^D(\varepsilon)}  \langle \xi , d\zeta(\pi(\alpha \wedge \alpha')) \xi' \rangle  \ \overline{\nu}^\varepsilon_1(d\xi) \overline{\nu}^\varepsilon_1(d\xi') \\
& \leq 2 \ N_\zeta +  2 \left\|  d\zeta(\pi(1)) \right\|  \leq 4 \ N_\zeta . 
\end{align*}
The estimate  in the last line is again from Lemma~\ref{lem:zetaprop}. 
The proof is concluded by Lemma~\ref{lem:basicappr}. 
\end{proof}

Proposition~\ref{prop:approxPF} is key in the following  lemma, which yields an alternative control over differences compared to just the Lipschitz continuity, Proposition~\ref{lem:LcontP}. The benefit here is the uniformity of the approximation $ D \to \infty $ as long as a norm $  \| \widehat x \|_q $ with  $ q > 1 $ of the integral kernel of $ x $ can be controlled. 
\begin{lemma}\label{lem:apprD2}
Let $ x \in \mathcal{S}_2^+ $ with integral kernel $ \widehat x $. For any $ \varepsilon > 0 $, any $ D \in \mathbb{N} $, all $ p \geq 2 $ and $ p^{-1} + q^{-1} = 1 $: 
$$
\sup_{\pi \in \Pi} \left| \mathcal{P}(\pi,x) - \mathcal{P}(\pi,Q_D x Q_D)  \right| \leq   2 \exp\left(4N_\zeta + \| \widehat x \|_1  - 2^D \varepsilon^2  \ln\left(\frac{ 2^D \varepsilon^2  }{\beta b  e }\right) \right)  + 4^{\frac{1}{q} } \varepsilon^{\frac{2}{p} }  \left\| \widehat x \right\|_q . 
$$
\end{lemma}
\begin{proof}
We pick $ \varepsilon > 0 $ and use the triangle inequality to estimate
\begin{align*}
& \left|   \mathcal{P}(\pi,x) - \mathcal{P}(\pi,Q_D x Q_D)   \right| \\
& \leq \left|  \mathcal{P}(\pi,x)  -  \mathcal{P}(\pi,x;D,\varepsilon) \right| +  \left|  \mathcal{P}(\pi,Q_D x Q_D )  -  \mathcal{P}(\pi,Q_D x Q_D ;D,\varepsilon)\right| +  \left|  \mathcal{P}(\pi,x;D,\varepsilon)  -  \mathcal{P}(\pi,Q_D x Q_D ;D,\varepsilon)  \right| .
\end{align*}
The first and second term on the right side is estimated using Proposition~\ref{prop:approxPF} and the bound
$\| \widehat{ (Q_D x Q_D)} \|_1 \leq \| \widehat x \|_1 $ 
on the $ L^1 $-norm integral of the  kernel of the projection, see also \eqref{eq:monnorm} below.  To bound the last term, we note that for  $ \xi \in S^D(\varepsilon) $ by the triangle inequality and a subsequent application of H\"older's inequality with $ p \geq 2 $ on $ L^2\left( (0,1)^2\right) $:
\begin{align*}
\left| \langle \xi , x \xi \rangle - \langle \xi , Q_D x Q_D  \xi \rangle \right|  & \leq  \left| \langle (1-Q_D) \xi , x \xi \rangle \right| +  \left| \langle Q_D \xi , x (1-Q_D) \xi \rangle \right| \\
& \leq \left( \| \xi \|_p + \| Q_D\xi \|_p\right)  \| (1-Q_D) \xi \|_p  \| \widehat x \|_q \\
&  \leq 2 \    2^{1-\frac{2}{p}} \ \| (1-Q_D) \xi \|_2^{\frac{2}{p}}  \| \widehat x \|_q  \leq 4^{\frac{1}{q} } \varepsilon^{\frac{2}{p}}  \| \widehat x \|_q . 
\end{align*}
This implies
$$
 \left| \mathcal{P}(\pi,x;D,\varepsilon)  -  \mathcal{P}(\pi,Q_D x Q_D ;D,\varepsilon)    \right| = 	\left|  X_0(\pi,Q_D x Q_D;D,\varepsilon)  -  X_0(\pi,x;D,\varepsilon) \right|   \leq 4^{\frac{1}{q} } \varepsilon^{\frac{2}{p}}  \| \widehat x \|_q ,
$$
which completes the proof.
\end{proof}

The next lemma shows that when restricting the outer supremum to a subset of the form
$
\mathcal{B}_q^+(R) \coloneqq \left\{ x \in \mathcal{S}_2^+ \ | \ \left\| \widehat x \right\|_q \leq R \right\}  $, 
where $ \widehat x $ is the integral kernel of $ x $, one may effectively pull the limit $ D \to \infty $ and obtain a convergence result similar to Lemma~\ref{lem:infconv}. 
\begin{lemma}\label{lem:maxminconv}
Let $ R > 0 $ and $ q \in (1,2] $ be arbitrary. Every sequence $ D \to \infty $ has a subsequence $ D_n \to \infty $ such that 
$$
\limsup_{D_n \to \infty} \sup_{x \in \mathcal{B}_q^+(R) } \left[ \inf_{\pi \in \Pi} \mathcal{P}\left(Q_{D_n}\pi Q_{D_n} , Q_{D_n} x Q_{D_n} \right) - \tfrac{1}{2} \xi^*\left( 2 Q_{D_n} x Q_{D_n} \right) \right] \leq \sup_{x \in \mathcal{S}_{2}^+ } \left[ \inf_{\pi \in \Pi} \mathcal{P}\left(\pi  , x \right) - \tfrac{1}{2} \xi^*\left( 2x \right) \right] 
$$ 
Moreover:
$ \quad\displaystyle
\liminf_{D \to \infty} \sup_{x \in \mathcal{S}_{2}^+ } \left[ \inf_{\pi \in \Pi} \mathcal{P}\left(Q_{D}\pi Q_{D} , Q_{D} x Q_{D} \right) - \tfrac{1}{2} \xi^*\left( 2 Q_{D} x Q_{D} \right) \right] \geq \sup_{x \in \mathcal{S}_{2}^+ } \left[ \inf_{\pi \in \Pi} \mathcal{P}\left(\pi  , x \right) - \tfrac{1}{2} \xi^*\left( 2x \right) \right] $.
\end{lemma} 
\begin{proof}
In this proof, we abbreviate the functional 
$$
\mathcal{K}(x) :=  \inf_{\pi \in \Pi} \mathcal{P}\left(\pi  ,  x\right) - \tfrac{1}{2} \xi^*\left( 2 x\right) , \quad x \in \mathcal{S}_{2}^+ .
$$

\noindent
\textit{Lower bound:}~~We let $ \varepsilon > 0 $, and pick $ x_\varepsilon \in \mathcal{S}_{2}^+ $ such that 
$
\mathcal{K}\left(x_\varepsilon \right)  \geq  \sup_{x \in \mathcal{S}_{2}^+ }\mathcal{K}(x) -\varepsilon $. 
This is possible since the supremum in the right side is finite by Proposition~\ref{prop:elpropP}.  The pre-limit in the right side of the assertion is trivially lower bounded for any $D $:
\begin{align}
\sup_{x \in \mathcal{S}_{2}^+ } \left[ \inf_{\pi \in \Pi} \mathcal{P}\left(Q_{D}\pi Q_{D} , Q_{D} x Q_{D} \right) - \tfrac{1}{2} \xi^*\left( 2 Q_{D} x Q_{D} \right) \right] & \geq \inf_{\pi \in \Pi} \mathcal{P}\left(Q_{D}\pi Q_{D} , Q_{D} x_\varepsilon Q_{D} \right) - \tfrac{1}{2} \xi^*\left( 2 Q_{D} x_\varepsilon Q_{D} \right) \notag \\
& \geq  \inf_{\pi \in \Pi}\mathcal{P}\left(Q_{D}\pi Q_{D} , Q_{D} x_\varepsilon Q_{D} \right) - \tfrac{1}{2} \xi^*\left( 2 x_\varepsilon \right) , \notag 
\end{align}
where the last step follows from Lemma~\ref{lem:relproj}, which guarantees $ \xi^*\left( 2 Q_{D} x_\varepsilon Q_{D} \right)  \leq \xi^*\left( 2 x_\varepsilon \right)  $. 
We now use Lemma~\ref{lem:infconv} to conclude 
\begin{equation}
\liminf_{D\to \infty} \inf_{\pi \in \Pi}\mathcal{P}\left(Q_{D}\pi Q_{D} , Q_{D} x_\varepsilon Q_{D} \right) - \tfrac{1}{2} \xi^*\left( 2 x_\varepsilon \right) =  \mathcal{K}\left(x_\varepsilon \right) ,
\end{equation}
and hence the claim, since $ \varepsilon > 0 $ was arbitrary.\\

\noindent
\textit{Upper bound:}~~We first fix $ D_0 \in \mathbb{N} $ and $ \varepsilon > 0 $, and suppose without loss of generality that $ D \geq D_0 $. We employ Lemma~\ref{lem:apprD2} and the fact that for any $ p \in [1,\infty) $ 
\begin{equation}\label{eq:monnorm}
 \| \widehat{ (Q_D x Q_D)} \|_p^p = \int_0^1 \int_0^1 \Big| \sum_{j,k=1}^{2^D} 2^{-2D} \int_{t_{j-1}}^{t_j} \int_{t_{k-1}}^{t_k} \widehat x(s',t') ds' dt' \ \indfct_{(t_{j-1},t_j]}(s)  \indfct_{(t_{k-1},t_k]}(t) \Big|^p ds dt  \leq \| \widehat x \|_p^p
\end{equation}
by Jensen's inequality. This guarantees that for any $ D \geq D_0 $:
\begin{align}\label{eq:restterms}
 \inf_{\pi \in \Pi} \mathcal{P}\left(Q_{D}\pi Q_{D} , Q_{D} x Q_{D} \right)  \leq  &  \inf_{\pi \in \Pi} \mathcal{P}\left(Q_{D}\pi Q_{D} , Q_{D_0} x Q_{D_0} \right) \notag \\
 & \qquad  +  2 \exp\left(4 N_\zeta + \| \widehat x \|_1  - 2^{D_0} \varepsilon^2  \ln\left(\frac{ 2^{D_0} \varepsilon^2  }{\beta b  e }\right) \right)  + 4^{\frac{1}{q} } \varepsilon^{\frac{2}{p} }  \left\| \widehat x \right\|_q  . 
\end{align}
Since $ x \in \mathcal{B}_q^+(R) $ and hence  $ \| \widehat x \|_1 \leq \| \widehat x \|_q \leq R $, the two terms in the last line are arbitrarily small if $ \varepsilon > 0 $ is small enough and $ D_0 > 0 $ is large enough. 
Since for $ D \geq D_0 $, we also have $ \xi^*\left( 2 Q_{D} x Q_{D} \right)  \geq \xi^*\left( 2 Q_{D_0} x Q_{D_0} \right)  $ by~Lemma~\ref{lem:relproj}.  It hence remains to bound the sequence of functionals
$$
\mathcal{K}_D(x^{D_0} ) \coloneqq  \inf_{\pi \in \Pi} \mathcal{P}\left(Q_{D}\pi Q_{D} , x^{D_0} \right)  - \frac{1}{2}  \xi^*\left( 2 x^{D_0}\right)  \quad\mbox{defined on $ x^{D_0} := Q_{D_0} x Q_{D_0} $.}
$$
By the finiteness of the supremum, for any $ \varepsilon > 0 $ and any $ D \geq D_0 $, there is $ x_\varepsilon^{D_0}(D) \in  \mathcal{B}_q^+(R) $ such that 
$$
 \sup_{ x \in \mathcal{B}_q^+(R)} \mathcal{K}_D(x^{D_0} )  \leq  \mathcal{K}_D( x_\varepsilon^{D_0}(D)) + \varepsilon  . 
 $$
Via $J_{D_0} $ in \eqref{def:JD}  and $I_{D_0} $ in \eqref{eq:QDrho}, the sequence $  \left( x_\varepsilon^{D_0}(D)  \right)_{D\geq D_0 } $ is bijectively identified with the sequence $  \left(   \big(x_\varepsilon^{D_0}(D)\big)_{k,k'}  \right) $ in $  \mathbb{R}^{2^{D_0}\times2^{D_0}} $. Moreover, since $  x_\varepsilon^{D_0}(D) \in  \mathcal{B}_q^+(R)  $, by~\eqref{eq:monnorm} we have
$$
\frac{1}{2^{2D_0} }\sum_{k,k'=1}^{2^{D_0}} \left|  \big(x_\varepsilon^{D_0}(D)\big)_{k,k'} \right|^q = \left\| \widehat {(Q_{D_0} x_\varepsilon(D) (Q_{D_0})} \right\|_q^q \leq  \left\| \widehat {x_\varepsilon(D)} \right\|_q^q \leq R^q .
$$
By compactness and the equivalence of norms in finite-dimensional vector spaces, there is some  $ \overline{x}_\varepsilon^{D_0} \in \mathcal{S}_{2,D_0}^+ $ and a subsequence $ D_n\to \infty $ such that 
$$ \lim_{n\to \infty} \| x_\varepsilon^{D_0}(D_n)  - \overline{x}_\varepsilon^{D_0} \|_2 = 0 . $$
We now pick $ \pi_\varepsilon(\overline{x}_\varepsilon^{D_0}) \in \Pi $ such that
$$
\inf_{\pi \in \Pi} \mathcal{P}\left(\pi , \overline{x}_\varepsilon^{D_0} \right)  \geq  \mathcal{P}\left( \pi_\varepsilon(\overline{x}_\varepsilon^{D_0})  , \overline{x}_\varepsilon^{D_0} \right)  - \varepsilon . 
$$ 
We have thus established the bound
$$
	 \sup_{ x \in \mathcal{B}_q^+(R)} \mathcal{K}_D(x^{D_0} )   \leq \mathcal{K}_D( x_\varepsilon^{D_0}(D)) + \varepsilon \leq   \mathcal{P}\left(Q_D \pi_\varepsilon(\overline{x}_\varepsilon^{D_0}) Q_D , x_\varepsilon^{D_0}(D) \right) - \frac{1}{2}  \xi^*\left( 2 x_\varepsilon^{D_0}(D) \right)  + \varepsilon , 
$$
with $ \varepsilon > 0 $ arbitrary, where in the second inequality the infimum is estimated by the choice $ \pi_\varepsilon(\overline{x}_\varepsilon^{D_0}) $. 
We now use the Lipschitz continuity, Proposition~\ref{lem:LcontP}, of the Parisi functional to further estimate 
\begin{align*}
 \mathcal{P}\left(Q_D \pi_\varepsilon(\overline{x}_\varepsilon^{D_0}) Q_D , x_\varepsilon^{D_0}(D) \right)  \leq  & \ \mathcal{P}\left(\pi_\varepsilon(\overline{x}_\varepsilon^{D_0}) , \overline{x}_\varepsilon^{D_0} \right) \\
 & \ + L \left( \| x_\varepsilon^{D_0}(D) - \overline{x}_\varepsilon^{D_0}  \|_2 + \int_0^1 \left\|Q_D \pi_\varepsilon(\overline{x}_\varepsilon^{D_0})  Q_D -  \pi_\varepsilon(\overline{x}_\varepsilon^{D_0})  \right\|_2 dm \right) . 
\end{align*}
Along the subsequence $ D_n \to \infty $, the last terms on the right vanish. By the lower semicontinuity of $ \zeta^* $, we also have
$$
 \liminf_{n\to \infty}  \xi^*\left( 2 x_\varepsilon^{D_0}(D_n) \right)  \geq  \xi^*\left( 2 \overline{x}_\varepsilon^{D_0}\right) . 
$$
Since $ \varepsilon > 0 $ was arbitrary, we thus arrive at the following bound on the sub-sequential upper limit:
$$
\limsup_{n\to \infty}  \sup_{ x \in \mathcal{B}_q^+(R)} \mathcal{K}_{D_n}(x^{D_0} ) \leq  \mathcal{K}(\overline{x}_\varepsilon^{D_0} ) \leq \sup_{x \in \mathcal{S}_2^+ }  \mathcal{K}(x) .
$$ 
The proof of the upper bound is completed by taking first $ D_0 \to \infty $ and then $ \varepsilon \downarrow 0 $ in the terms on the second line of~\eqref{eq:restterms}.
\end{proof}

\section{Proof of the main results}\label{sec:PMain}
Our proof of Theorems~\ref{thm:main}, \ref{thm:Pan} and~\ref{thm:mainsc} is based on Chen's versions \cite{Chen23,Chen23b} of a Parisi formula and the approximation results in the previous subsections. We start with the proof of the simplest case, whose Corollary~\ref{cor:differble} we will partially use in the proof of the second identity in Theorem~\ref{thm:main}.

\subsection{Proof of Theorem~\ref{thm:mainsc}}

\begin{proof}[Proof of Theorem~\ref{thm:mainsc}]
By Proposition~\ref{prop:approx} it is enough to investigate the limit $ N \to \infty $ of  $ W_N^D(x)  $ and take a subsequent limit $ D \to \infty $.  Since  $ W_N^D(x)  $ coincides with the self-overlap corrected partition function of a vector spin glass for fixed $ D \in \mathbb{N} $, $ x \in \mathcal{S}_2 $, 
Chen's  result  \cite{Chen23} yields
\begin{align*}
\lim_{N\to \infty} \frac{1}{N} \mathbb{E}\left[ \ln W_N^D(x)   \right] = \inf_{\pi \in \Pi_D } \mathcal{P}_D(\pi, J_D(x)) . 
\end{align*}
The infimum is over paths 
$$ \Pi_D = \left\{ \pi: [0,1) \to \mathcal{D}_D  \; \mbox{c\`adl\`ag} \  \big| \  0 \leq \pi(s) \leq \pi(t) \; \mbox{for all $ s \leq t $ with }\; \pi(0) = 0   \right\}  $$
and for any $ \pi \in \Pi_D $ and $ x \in  \mathcal{S}_{2} $
\begin{multline}
\mathcal{P}_D(\pi, J_D(x)) = \\  \mathbb{E}\left[ \ln \sum_{\alpha \in \mathbb{N}^{r} }\mu^\pi_\alpha \int \exp\left( \langle w_\alpha^D, \sigma \rangle_D +  \langle \sigma , \left(J_D(x) - \tfrac{1}{2} d\zeta_D(\pi(1)) \right) \sigma \rangle_D \right) \nu_{1}^D(d\sigma)\right] + \frac{1}{2} \int_0^1 \theta_D\left(\pi(m) \right) dm .
\end{multline}
The measure $ \nu_{1}^D $ refers to the push-forward of the measure $ \nu_{1} $ which is supported on $ \mathbb{R}^{2^D}$-vectors  of the form~\eqref{eq:defsigmaD}. 
The expectation refers to the average over a centred, $ \mathbb{R}^{2^D} $-valued Gaussian process $ w_\alpha $, whose covariance is 
\begin{align*}
 \mathbb{E}\left[ \langle w_\alpha, \sigma\rangle_D  \langle w_{\alpha'}, \sigma'\rangle_D  \right]  & = \langle \sigma , d\zeta_D(\pi(\alpha \wedge \alpha') \sigma' \rangle_D = \frac{1}{2^{2D}} \sum_{k,k'=1}^{2^D}  \sigma_k  \sigma'_{k'} \ d\zeta_D\left(\pi(\alpha \wedge \alpha')\right)_{k,k'} \\
 & =  d\zeta\left(I_D\left( \pi(\alpha \wedge \alpha') \right) \right) I_D(|\sigma\rangle \langle \sigma' | )  . 
\end{align*}
By definition of the push-forward measure, we may hence identify $ \mathcal{P}_D(\pi, J_D(x)) = \mathcal{P}\left(I_D(\pi), Q_D x Q_D \right)  $ for any $\pi \in  \Pi_D $ and $ x \in \mathcal{S}_{2} $, such that
$$
\inf_{\pi \in \Pi_D } \mathcal{P}_D(\pi, J_D(x)) = \inf_{\pi \in \Pi} \mathcal{P}(Q_D\pi Q_D,   Q_D x Q_D  ) . 
$$
The proof is concluded by an application of Lemma~\ref{lem:infconv}. 
\end{proof} 

\begin{remark}
Alternatively to the above approximation argument, one could have used a straightforward extension of Chen's version~\cite[Sec.~2]{Chen23} of Guerra's interpolation bound~\cite{Guerra:2003aa} for the path integral $ W_N(x) $ to directly establish the upper bound $ \limsup_{N\to \infty} \frac{1}{N} \mathbb{E}\left[ \ln W_N(x)   \right] \leq \inf_{\pi \in \Pi_D } \mathcal{P}(\pi, x) $ for any $ x \in \mathcal{S}_2 $ without any approximation with square-wave pulses. 
\end{remark}

\subsection{Proof of Corollary~\ref{cor:differble}}\label{sec:Cor1.9}

\begin{proof}[Proof of Corollary~\ref{cor:differble}]
Since $ \mathcal{S}_2 \ni x \mapsto \mathcal{P}(\pi,x) $ is $ 1 $-Lipschitz continuous by~\eqref{eq:LcontP}, the $ 1 $-Lipschitz continuity of $ \mathcal{F}(x) = \inf_{\pi\in \Pi} \mathcal{P}(\pi,x) $ is immediate. The convexity follows from Theorem~\ref{thm:mainsc}, which identifies $ \mathcal{F} $ as a pointwise limit of convex functions. The bounds in Proposition~\ref{prop:elpropP}  show that  $ \dom \mathcal{F} =  \mathcal{S}_2  $. \\[1ex]
\noindent
1.~~The Lengendre transform $ \mathcal{F}^* $ is non-negative, since
$$
\mathcal{F}^*(\varrho) \geq -  \mathcal{F}(0) \geq 0  
$$
by Proposition~\ref{prop:elpropP}. 

If $ \varrho \not\in \overline{\mathcal{D}} $, then by the Hahn-Banach separation theorem there is some $ x_0 \in \mathcal{S}_2 $ and $ s \in \mathbb{R} $ such that $  \langle x_0 ,\varrho\rangle  \leq  s < \langle x_0 , \varrho' \rangle $ for all $ \varrho'  \in \overline{\mathcal{D}} $. Since $\supp \nu_1  \subset \overline{\mathcal{D}}$, we thus have
$$
\sup_{\xi \in \supp \nu_1}  \langle x_0 , (| \xi\rangle\langle \xi | - \varrho)\rangle \leq s -   \langle x_0 ,\varrho\rangle   < 0 . 
$$ 
With the help of Proposition~\ref{prop:elpropP} and by estimating the supremum from below by picking $ x = R x_0 $ with $ R > 0 $ arbitrary, we thus conclude that 
$$
 \mathcal{F}^*(\varrho) \geq \sup_{x \in \mathcal{S}_2} \left( \langle x , \varrho \rangle - \mathcal{P}(0,x)  \right) \geq - \ln \int \exp\left( R \langle x_0 , (| \xi\rangle\langle \xi | - \varrho)\rangle \right) \nu_1(d\xi) \geq R ( s -   \langle x_0 ,\varrho\rangle) . 
 $$
Hence $  \mathcal{F}^*(\varrho) = \infty$, since $ R > 0 $ is arbitrary. This proves $ \dom  \mathcal{F}^* \subset  \overline{\mathcal{D}} = \mathcal{D} $, where the last equality reflects the fact that $ \mathcal{D} $ is closed. 

The fact that $  \mathcal{F} = \left(  \mathcal{F}^*\right)^* $ is  the Fenchel–Moreau theorem~\cite{BauCom17}. The latter also guarantees that $ \mathcal{F}^* $ is a proper functional. \\

\noindent
2.~~Since $ \mathcal{F} $  is a proper, convex  and continuous (which is essential to know in the infinite-dimensional situation) functional, the Gateaux differentiability of $ \mathcal{F} $  at $ x \in \mathcal{S}_2 $ follows \cite[Prop.~17.26]{BauCom17} from the fact that the set of subdifferentials $\partial \mathcal{F}(x) $ only contains one point, $ \partial \mathcal{F}(x) =\{ u \} $, which then equals $ d\mathcal{F}(x) $. Recall \cite{BauCom17} that $ u \in \partial \mathcal{F}(x) $ if and only if for any $ y \in \mathcal{S}_2 $:
\begin{equation}\label{eq:subdifferentiable}  
 \mathcal{F}(x) \leq  \mathcal{F}(y) + \langle y-x , u \rangle .
\end{equation}
The argument then proceeds as in~\cite[Proof of Prop.~2.3]{Chen23b} or \cite[Proof of Thm.~3.7]{Pan13}. We will therefore only sketch it. For fixed $ x \in \mathcal{S}_2 $, and $ n \in \mathbb{N} $, one picks $ \pi_n \in \Pi $ such that 
\begin{equation}\label{eq:Festinf}
\mathcal{F}(x) \leq \mathcal{P}(\pi_n,x) \leq \mathcal{F}(x)  + \frac{1}{n} .
\end{equation}
which by~\eqref{eq:subdifferentiable} implies that for all $ y \in \mathcal{S}_2 $, $ \lambda \in (0,1] $ and all $ n, m \in \mathbb{N} $:
\begin{equation}\label{eq:convexitysubdiff}
\frac{ \mathcal{P}(\pi_m,x) -  \mathcal{P}(\pi_m,x-\lambda y) - 1/m }{\lambda} \leq  \langle y, u \rangle \leq \frac{  \mathcal{P}(\pi_n,x+\lambda y)   - \mathcal{P}(\pi_n,x) + 1/n }{\lambda} . 
\end{equation}
From Proposition~\ref{prop:elpropP}, one concludes the Taylor estimate for any $ \pi \in \Pi $:
$$ \left|  \mathcal{P}(\pi ,x\pm\lambda y)  -  \mathcal{P}(\pi,x)  \mp \lambda \langle y, d\mathcal{P}(\pi,x) \rangle \right| \leq  \lambda^2 \| y\|^2 ,
$$
which upon insertion in~\eqref{eq:convexitysubdiff} then yields to the inequalities
$$
 \left|  \langle y, d\mathcal{P}(\pi_m,x) \rangle -  \langle y, d\mathcal{P}(\pi_n,x) \rangle\right| \leq   2 \lambda \| y\|^2 + \frac{2}{\min\{n,m\} \lambda}  , \quad  \left|  \langle y, u \rangle -  \langle y, d\mathcal{P}(\pi_n,x) \rangle \right| \leq  \lambda \| y\|^2 + \frac{1}{n \lambda} .
$$ 
Setting $  \lambda  = \min\{n,m\}^{-1/2} $ respectivley $ \lambda = n^{-1/2} $, the above bounds then show that $ \lim_{n\to \infty}   \langle y, d\mathcal{P}(\pi_n,x) \rangle $ converges. Moreover, the limit is the only point in $ \partial \mathcal{F} $. 

The convergence~\eqref{eq:fvconv} then follows from the convexity of $ G_N $, which implies that for all $ x , y \in \mathcal{S}_2 $ and $ \lambda \in (0,1] $:
$$
\frac{G_N(x) - G_N(x-\lambda y)}{\lambda} \leq \langle y , dG_N(x) \rangle \leq  \frac{G_N(x+\lambda y)-G_N(x)}{\lambda} ,
$$
together with the pointwise convergence~\eqref{eq:mainsc} and the differentiability of the limit. 

The proof of the concentration~\eqref{eq:concentration} proceeds exactly as in~\cite[Proof of Prop. 2.4]{Chen23b}.  In a first step, one establishes the convergence
\begin{equation}\label{eq:conc1a}
\lim_{N\to \infty} \mathbb{E}\left[  \left\langle \big|  Y_N - \langle  Y_N\rangle_{x,N}  \big| \right\rangle_{x,N} \right] =0 \qquad \mbox{with $ Y_N(\bxi) \coloneqq   \langle   \langle r_N( \bxi ,  \bxi ),  y \rangle $.}
\end{equation}
This is based on the bound valid for any $ r > 0 $:
\begin{align*}
	& \mathbb{E}\left[ \left\langle \big|  Y_N(\bxi^{(1)})  -  Y_N(\bxi^{(2)}) \big| \right\rangle_{x,N}  \right] \\ & = \frac{1}{r} \int_0^r  \mathbb{E}\left[\left\langle \big|  Y_N(\bxi^{(1)})  -  Y_N(\bxi^{(2)}) \big| \right\rangle_{x+sy,N}  \right] ds -  \frac{1}{r}  \int_0^r  \int_0^t \frac{d}{ds}  \mathbb{E}\left[\left\langle \big|  Y_N(\bxi^{(1)})  -  Y_N(\bxi^{(2)}) \big| \right\rangle_{x+sy,N}  \right]  ds dt \\
	 & \leq  \frac{1}{r} \int_0^r  \left(  \mathbb{E}\left[\left\langle \big|  Y_N(\bxi^{(1)})  -  Y_N(\bxi^{(2)}) \big|^2 \right\rangle_{x+sy,N}  \right]\right)^{1/2}  \mkern-5mu ds + \frac{8 N}{r}  \int_0^r  \int_0^t  \mathbb{E}\left[\left\langle \big|  Y_N(\bxi^{(1)})  -  Y_N(\bxi^{(2)}) \big|^2 \right\rangle_{x+sy,N}  \right]  ds dt \\
	 & \leq \sqrt{\frac{\varepsilon_N(r)}{r N} } + 8 \varepsilon_N(r) , \quad\mbox{with}\quad \varepsilon_N(r)\coloneqq N \mathbb{E}\left[ \langle y , dG(x+ry) - dG_N(x)  \rangle  \right] . 
\end{align*}
Convexity of $ G_N(x) $ allows to bound
\begin{align}\label{eq:lastconvexitybd}
\limsup_{N\to \infty} \varepsilon_N(r) & \leq \limsup_{N\to \infty} \left( \frac{G_N(x+(r+t) y) - G_N(x+ry)}{t} -  \frac{G_N(x) - G_N(x-ty)}{t} \right) \notag \\
&  =  \frac{\mathcal{F}(x+(r+t) y) - \mathcal{F}(x+ry)}{t} -  \frac{\mathcal{F}(x) - \mathcal{F}(x-ty)}{t}
\end{align}
for any $ t > 0 $, where the equality is by~\eqref{eq:mainsc}. By the differentiability of $ \mathcal{F} $ at $ x $, the right side converges to zero if one first takes the limit  $ r \downarrow 0 $ and then $ t \downarrow 0 $. This proves~\eqref{eq:conc1a}.

The second step rests on the concentration of the free energy, stated in Proposition~\ref{Prop:self-averaging}, which addresses 
$$ \Phi_N(x) \coloneqq N^{-1} \ln W_N(x) ,
$$  together with the basic convexity bounds for 
$ \langle  Y_N\rangle_{x,N}  - \mathbb{E}\left[  \langle  Y_N\rangle_{x,N}\right]  =\langle y , d \Phi_N(x) - dG_N(x) \rangle $. More precisely, for any $ t > 0 $  with $ \delta_N(t) \coloneqq \left|  \Phi_N(x+ t y) - G_N(x+ty) \right| + \left|  \Phi_N(x- t y) - G_N(x-ty) \right| + 2 \left|  \Phi_N(x) - G_N(x) \right| $, we have
\begin{equation}\label{eq:fluctderivative}
\mathbb{E}\left[ \left| \langle y , d \Phi_N(x) - dG_N(x) \rangle\right| \right] \leq \frac{1}{t}\left[G_N(x+ty) + G_N(x-ty) - 2 G_N(x)  + \mathbb{E}\left[\delta_N(t) \right] \right] .
\end{equation}
From Proposition~\ref{Prop:self-averaging} one concludes
$ \limsup_{N\to \infty} \sup_{t\in[0,1] } \mathbb{E}\left[\delta_N(t) \right] = 0 $. By the same argument as in~\eqref{eq:lastconvexitybd}, one thus shows that the upper limit $ N \to \infty $ of the right side in~\eqref{eq:fluctderivative} 
converges to zero as $ t \downarrow 0 $. 
\end{proof}
 
 \subsection{Proof of Theorems~\ref{thm:main} and~\ref{thm:Pan}}

The proof of the main results, Theorems~\ref{thm:main} and~\ref{thm:Pan}, is  more involved than that of Theorem~\ref{thm:mainsc} due to the presence of the outer supremum in their statements. 
Before we spell the approximation argument, we establish the representation~\eqref{eq:TolandSinger}  and through it, the fact that the outer supremum in~\eqref{eq:thmmain} may be restricted to a sufficiently large ball. 
\begin{lemma}\label{lem:TolandSinger}
Under Assumption~\ref{ass}:
\begin{enumerate}
\item
for $  \mathcal{F}(x) = \inf_{\pi \in \Pi} \mathcal{P}\left(\pi , x  \right) $, $ x \in \mathcal{S}_2 $, and its Legendre transform $ \mathcal{F}^*(\varrho) = \sup_{x \in \mathcal{S}_2} \left( \langle x , \varrho \rangle -   \mathcal{F}(x)  \right) $:
\begin{equation}\label{eq:TS2}
 \sup_{x \in \mathcal{S}_2^+} \left[ \mathcal{F}\left(x \right) - \frac{1}{2} \zeta^*(2x)  \right]  = 
 \sup_{\varrho \in \mathcal{D} } \left[ \frac{1}{2} \zeta(\varrho)  -  \mathcal{F}^*\left(\varrho  \right)\right]  .
\end{equation} 
The supremum in the left side agrees with the supremum restricted to $ B_2(N_\zeta) = \left\{ x \in \mathcal{S}_2^+ | \ \| x \|_2 \leq N_\zeta \right\}$.

\item for any $ D \in \mathbb{N} $ let $  \mathcal{F}_D(x) = \inf_{\pi \in \Pi} \mathcal{P}\left(Q_D \pi Q_D , Q_D x  Q_D \right) $, $ x \in \mathcal{S}_2 $, and set its Legendre transform $ \mathcal{F}_D^*(\varrho) = \sup_{x \in \mathcal{S}_2} \left( \langle  Q_D x Q_D , \varrho \rangle -   \mathcal{F}_D(x)  \right) $. One has $  \dom\mathcal{F}_D^* \subset \mathcal{D}^D \coloneqq \{ Q_D \varrho Q_D \ | \ \varrho \in \mathcal{D} \} $, and 
\begin{equation}\label{eq:TS3}
 \sup_{x \in \mathcal{S}_2^+} \left[ \mathcal{F}_D\left(x \right) - \frac{1}{2} \zeta^*(2Q_D x Q_D)  \right]  = 
 \sup_{\varrho \in \mathcal{D} } \left[ \frac{1}{2} \zeta(Q_D\varrho Q_D)  -  \mathcal{F}_D^*\left(\varrho  \right)\right]  .
\end{equation}
The supremum in the left side agrees with the supremum restricted to $ B_2(N_\zeta)  $.
\end{enumerate}
\end{lemma} 
\begin{proof}
The representation~\eqref{eq:TolandSinger} follows from the Toland-Singer theorem~\cite[Cor. 14.20]{BauCom17}, which ensures the equivalence of the dual extremal problems:
\begin{equation}\label{eq:TolandSingerabs}
\sup_{2 x \in \dom \zeta^*} \left( \left(\mathcal{F}^*\right)^*(x) - \tfrac{1}{2} \zeta^*(2x) \right) = \sup_{\varrho \in \dom \mathcal{F}^* } \left( \tfrac{1}{2} \zeta(\varrho) - \mathcal{F}^*(\varrho) \right) .
\end{equation}
By Corollary~\ref{cor:differble} one has $ \left(\mathcal{F}^*\right)^* = \mathcal{F} $. Since also $ \dom \zeta^* = \mathcal{S}_2^+ $, the left side in~\eqref{eq:TolandSingerabs} agrees with the right side in~\eqref{eq:TS2}. Moreover, since $  \dom \mathcal{F}^* \subset \mathcal{D}$ by Corollary~\ref{cor:differble}, and $  \mathcal{D} \subset \dom \zeta $ by Lemma~\ref{lem:zetaprop}, the supremum in the right side of~\eqref{eq:TolandSingerabs} agrees with the one over $ \varrho \in \mathcal{D} $.  This completes the proof of~\eqref{eq:TS2}. 

To show that one may  restrict the supremum in the left side of \eqref{eq:TS2} to the ball $ B_2(N_\zeta) $, we fix $ \varepsilon > 0 $ arbitrary. Since the supremum in the right side is finite by Proposition~\ref{prop:elpropP}, we may choose $ \varrho_\varepsilon \in \mathcal{D}$ such that 
$$
 \sup_{\varrho \in \mathcal{D}} \left( \tfrac{1}{2} \zeta(\varrho) - \mathcal{F}^*(\varrho) \right) \leq \tfrac{1}{2} \zeta(\varrho_\varepsilon ) - \mathcal{F}^*(\varrho_\varepsilon ) + \varepsilon .
$$
%Since $ \dom \mathcal{F}^* \subset \overline{\mathcal{D}} $ and $ \zeta $ is weakly lower semicontinuous, there is $ \varrho_\varepsilon' \in \mathcal{D} $ such that i)~$ \|  \varrho_\varepsilon'  -  \varrho_\varepsilon \|_2 < 2 \varepsilon/ N_\zeta $ and ii)~$  \zeta(\varrho_\varepsilon )  \leq  \zeta(\varrho_\varepsilon' ) + \varepsilon $. 
We now pick
\begin{equation}\label{eq:xfromrho}
2 x_\varepsilon \coloneqq d  \zeta(\varrho_\varepsilon) ,
\end{equation}
where $ d  \zeta(\varrho_\varepsilon) $ is the Gateaux derivative at $  \varrho_\varepsilon \in \mathcal{D} $, which by Lemma~\ref{lem:zetaprop} is in $ \mathcal{S}_2^+ $ and bounded 
$$
\| x_\varepsilon  \|_2 \leq \frac{N_\zeta}{2} .
$$
The differentiability also implies \cite[Prop. 16.9]{BauCom17} that $ \frac{1}{2}\left(\zeta(  \varrho_\varepsilon) + \zeta^*(2  x_\varepsilon)  \right) =  \frac{1}{2}  \langle 2 x_\varepsilon,  \varrho_\varepsilon \rangle =  \langle x_\varepsilon , \varrho_\varepsilon \rangle $.  Hence
\begin{align}\label{eq:upperbdxball}
 \sup_{\varrho \in \mathcal{D}} \left( \tfrac{1}{2} \zeta(\varrho) - \mathcal{F}^*(\varrho) \right) &  \leq  \tfrac{1}{2} \zeta(\varrho_\varepsilon ) - \langle \varrho_\varepsilon , x_\varepsilon\rangle +  \mathcal{F}(x_\varepsilon ) + \varepsilon  = \mathcal{F}(x_\varepsilon ) -   \tfrac{1}{2} \zeta^*(2 x_\varepsilon ) + \varepsilon .  %\notag \\
% & \leq   \tfrac{1}{2} \zeta(\varrho_\varepsilon' )  - \langle \varrho_\varepsilon' , x_\varepsilon'\rangle +  \mathcal{F}(x_\varepsilon' ) +  \|  \varrho_\varepsilon'  -  \varrho_\varepsilon \|_2 \| x_\varepsilon'  \|_2 + 2 \varepsilon \notag \\
 %& = \mathcal{F}(x_\varepsilon ) -   \tfrac{1}{2} \zeta^*(2 x_\varepsilon ) + \varepsilon . 
\end{align}
From~\eqref{eq:TS2} and since $ \varepsilon > 0 $  was arbitrary, we thus conclude that the supremum in the left side of~\eqref{eq:TS2} agrees with the one restricted to the ball $ B_2(N_\zeta) \ni x_\varepsilon $. \\

The claim concerning the domain of $ \mathcal{F}_D^* $ is established along the lines of the corresponding result in Corollary~\ref{cor:differble}. Note that $ \mathcal{D}^D $ is closed, $ \mathcal{D}^D = \overline{\mathcal{D}^D}$, since $ \mathcal{D} $ is closed. From Proposition~\ref{prop:elpropP}  one concludes the lower bound
$$
 \mathcal{F}_D^*(\varrho) \geq  - \ln \int \exp\left( \langle x , Q_D \left( |\xi\rangle \langle \xi| -  \varrho \right) Q_D \rangle \right) \nu_1(d\xi) ,
$$
for any $ x \in \mathcal{S}_2 $. If $ \varrho \not\in \overline{\mathcal{D}^D} $, then by a Hahn-Banach separation argument, there is $ x_0 \in \mathcal{S}_2 $ and $ s \in \mathbb{R} $ such that $ \sup_{\xi \in \supp \nu_1}  \langle x_0 , Q_D (| \xi\rangle\langle \xi | - \varrho) Q_D \rangle \leq s -   \langle x_0 , Q_D\varrho Q_D \rangle   < 0  $. Picking again $ x = R x_0 $ with $ R > 0 $ arbitrarily large, we conclude $ \dom\mathcal{F}_D^* \subset  \overline{\mathcal{D}^D}   $.
The Toland-Singer theorem then implies~\eqref{eq:TS3}  with the supremum taken over $ \varrho \in \dom \mathcal{F}_D^* $ instead of $  \varrho \in \mathcal{D} $. Since $ \dom \mathcal{F}_D^*  \subset \overline{\mathcal{D}^D} = \mathcal{D}^D $ and $ \zeta(Q_D \varrho Q_D) < \infty $ for any $ \varrho \in \mathcal{D} $, we then arrive at~\eqref{eq:TS3}. 

To see that the supremum in the left side  can be restricted to $B_2(N_\zeta) $, we proceed analogously as above. Note that by definition $ \mathcal{F}^*_D(\varrho) =  \mathcal{F}^*_D(Q_D \varrho Q_D ) $ for any $ \varrho \in \dom  \mathcal{F}^*_D $. For $ \varepsilon > 0 $ we first pick $ \varrho_\varepsilon \in \mathcal{D} $ such that $Q_D \varrho_\varepsilon Q_D \in \mathcal{D}^D $ and
$$
 \sup_{\varrho \in \mathcal{D} } \left( \tfrac{1}{2} \zeta(Q_D\varrho Q_D) - \mathcal{F}_D^*(\varrho ) \right) \leq \tfrac{1}{2} \zeta(Q_D\varrho_\varepsilon Q_D) - \mathcal{F}_D^*(Q_D \varrho_\varepsilon Q_D) + \varepsilon .
$$
%We then pick $ \varrho_\varepsilon' \in \mathcal{\mathcal{D}^D} $ such that i)~$ \|  \varrho_\varepsilon'  -  Q_D \varrho_\varepsilon Q_D \|_2 < 2 \varepsilon/ N_\zeta $ and ii)~$  \zeta(Q_D\varrho_\varepsilon Q_D )  \leq  \zeta(\varrho_\varepsilon' ) + \varepsilon $, and set 
We then set $
2 x_\varepsilon \coloneqq d  \zeta(Q_D\varrho_\varepsilon Q_D) $ similarly as in~\eqref{eq:xfromrho} such that  $ x_\varepsilon \in B_2(N_\zeta) $ and $ x_\varepsilon = Q_D x_\varepsilon Q_D $ (cf.\ Lemma~\ref{lem:relproj}). An identical argument as in~\eqref{eq:upperbdxball} yields
$$
 \sup_{x \in \mathcal{S}_2^+} \left[ \mathcal{F}_D\left(x \right) - \frac{1}{2} \zeta^*(2Q_D x Q_D)  \right]  \leq    \mathcal{F}_D(x_\varepsilon) -   \tfrac{1}{2} \zeta^*(2 Q_D x_\varepsilon Q_D) +  \varepsilon . 
$$
This completes the proof.
\end{proof} 
We are finally ready to spell the proof of the main results.  
\begin{proof}[Proof of Theorem~\ref{thm:main}] 
Since the representation~\eqref{eq:TolandSinger} is proven in Lemma~\ref{lem:TolandSinger}, 
it remains to establish the convergence~\eqref{eq:thmmain}. To do so, Proposition~\ref{prop:approx} allows us to approximate the $ N \to \infty $ limit  by
\begin{equation}\label{eq:Chen1}
\lim_{N\to \infty} \frac{1}{N} \mathbb{E}\left[ \ln Z_N^{D}   \right]  = \sup_{x \in \mathcal{S}_{2}^+ } \left[ \inf_{\pi \in \Pi_D } \mathcal{P}_{D}(\pi, J_D(x)) - \tfrac{1}{2} \xi_{D}^*(2J_D(x)) \right] ,
\end{equation}
where the equality is \cite[Thm.~1.2]{Chen23}, and $ \xi_{D} $ and $  \xi_{D}^* $ are defined as in~\eqref{eq:defzetaD} and \eqref{def:LenegdreD}  respectively using $ \widehat\zeta $ of Assumption~\ref{ass}.  Since $ x \in \mathcal{S}_2^+ $ is projected via $ J_D $ onto the non-negative $ 2^D\times 2^D $ matrices $ \mathcal{S}_{2,D}^+ $, the supremum is effectively over this set.  The infimum is over paths 
$$ \Pi_D = \left\{ \pi: [0,1) \to \mathcal{D}_D  \; \mbox{c\`adl\`ag} \  \big| \  0 \leq \pi(s) \leq \pi(t) \; \mbox{for all $ s \leq t $ with }\; \pi(0) = 0   \right\}  $$
and for any $ \pi \in \Pi_D $ and $ x \in  \mathcal{S}_{2,D} $
\begin{multline}\label{eq:VGParisiF}
\mathcal{P}_{D}(\pi, x) =  \mathbb{E}\left[ \ln \sum_{\alpha \in \mathbb{N}^{r} }\mu^\pi_\alpha \int \exp\left( \langle w_\alpha , \sigma \rangle_D + \langle \sigma , \left(x - \tfrac{1}{2} d\zeta_D(\pi(1)) \right) \sigma \rangle_D \right) \nu_1^D(d\sigma)\right] + \frac{1}{2} \int_0^1 \theta_D\left(\pi(s) \right) ds .
\end{multline}
The measure $ \nu_1^D $ refers to the push-forward of the measure $ \nu_1 $ which is supported on $ \mathbb{R}^{2^D}$-vectors  of the form~\eqref{eq:defsigmaD}. 
The expectation is the average over a centred, $ \mathbb{R}^{2^D} $-valued Gaussian process $ w_\alpha $ with $\alpha \in \mathbb{N}^{r} $, whose covariance is 
\begin{align*}
 \mathbb{E}\left[ \langle w_\alpha , \sigma\rangle  \langle w_{\alpha'} , \sigma'\rangle  \right]  & = \langle \sigma , d\zeta_{D}(\pi(\alpha \wedge \alpha') \sigma' \rangle_D = \frac{1}{2^{2D}} \sum_{k,k'=1}^{2^D}  \sigma_k  \sigma'_{k'} \ d\zeta_{D}\left(\pi(\alpha \wedge \alpha')\right)_{k,k'} \\
 & =  d\zeta\left(I_D\left( \pi(\alpha \wedge \alpha') \right) \right) I_D(|\sigma\rangle \langle \sigma' | )  . 
\end{align*}
By definition of the push-forward measure, Lemma~\ref{lem:relproj} and \eqref{eq:defthetaD}, we may hence identify for any $\pi \in  \Pi_D $ and $ x \in \mathcal{S}_{2} $:
$$
\mathcal{P}_{D}(\pi, J_D(x)) = \mathcal{P}\left(I_D(\pi), Q_D x Q_D \right) .
$$ 
Based on this and another application of Lemma~\ref{lem:relproj}, we may rewrite
\begin{equation}\label{eq:supinfD}
 \sup_{x \in \mathcal{S}_{2}^+ } \left[ \inf_{\pi \in \Pi_D } \mathcal{P}_{D}(\pi, J_D(x)) - \tfrac{1}{2} \xi_{D}^*(2 J_D(x)) \right]   =  \sup_{x \in \mathcal{S}_{2}^+ } \left[ \inf_{\pi \in \Pi} \mathcal{P}\left(Q_D\pi Q_D , Q_D x Q_D \right) - \tfrac{1}{2} \xi^*\left( 2 Q_D x Q_D \right) \right] . 
\end{equation}
To control the subsequent limit $ D \to \infty $, we derive matching lower and upper bounds from Lemma~\ref{lem:maxminconv}.\\

For the lower bound, Lemma~\ref{lem:maxminconv} together with the above consequence \eqref{eq:Chen1} of Chen's result yields
\begin{align*}
 \liminf_{D\to\infty}\liminf_{N\to \infty} \frac{1}{N} \mathbb{E}\left[ \ln Z_N^{D}   \right] & = \liminf_{D\to\infty} \sup_{x \in \mathcal{S}_{2}^+} \left[ \inf_{\pi \in \Pi} \mathcal{P}\left(Q_D\pi Q_D , Q_D x Q_D \right) - \tfrac{1}{2} \xi\left( 2 Q_D x Q_D \right) \right] \\ &  \geq \sup_{x \in \mathcal{S}_{2}^+ } \left[ \inf_{\pi \in \Pi} \mathcal{P}\left(\pi , x  \right) - \tfrac{1}{2} \xi\left( 2  x \right) \right] . 
\end{align*}

For the upper bound, we first use Lemma~\ref{lem:TolandSinger} to restrict the supremum to $ \mathcal{B}^+_2(N_\zeta) = \left\{ x \in \mathcal{S}_2^+ | \ \| x \|_2 \leq N_\zeta \right\}$,
$$
\sup_{x \in \mathcal{S}_{2}^+ }\! \left[ \inf_{\pi \in \Pi} \mathcal{P}\left(Q_D\pi Q_D , Q_D x Q_D \right) - \tfrac{1}{2} \xi^*\left( 2 Q_D x Q_D \right) \right]  = \mkern-5mu \sup_{x \in \mathcal{B}^+_2(N_\zeta) } \!\left[ \inf_{\pi \in \Pi}\mathcal{P}\left(Q_D\pi Q_D , Q_D x Q_D \right) - \tfrac{1}{2} \xi^*\left( 2 Q_D x Q_D \right) \right] .
$$ 
Then Lemma~\ref{lem:maxminconv} with $ q = 2 $ ensures that there is a sequence $ D\to \infty $ such that 
\begin{align*}
\limsup_{D\to\infty}\limsup_{N\to \infty} \frac{1}{N} \mathbb{E}\left[ \ln Z_N^{D}    \right] 
 & =   \limsup_{D\to\infty} \sup_{x \in B_2(N_\zeta)  }  \left[ \inf_{\pi \in \Pi} \mathcal{P}\left(Q_{D}\pi Q_{D} , Q_{D} x Q_{D} \right) - \tfrac{1}{2} \xi^*\left( 2 Q_{D} x Q_{D} \right) \right]  \\
&  \leq  \sup_{x \in \mathcal{S}_{2}^+ } \left[ \inf_{\pi \in \Pi} \mathcal{P}\left(\pi , x  \right) - \tfrac{1}{2} \xi\left( 2  x \right) \right] . 
\end{align*} 
This completes the proof of~\eqref{eq:thmmain}.
\end{proof} 

We finally spell the proof of the Theorem~\ref{thm:Pan}.
\begin{proof}[Proof of Theorem~\ref{thm:Pan}] 
The set inclusion $ \Pi_\varrho \subset \Pi $ for any endpoint $ \varrho \in \mathcal{D} $, implies through the representation~\eqref{eq:TolandSinger} the upper bound
$$
\limsup_{N\to \infty} \frac{1}{N} \mathbb{E}\left[ \ln Z_N   \right] \leq   \sup_{\varrho \in \mathcal{D} } \left[\inf_{x \in \mathcal{S}_2} \inf_{\pi \in \Pi_{\varrho}}  \mathcal{P}\left(\pi, x \right) - \langle x , \varrho  \rangle + \frac{1}{2} \zeta(\varrho)  \right] \eqqcolon F  .
$$
Since the right side is finite by Proposition~\ref{prop:elpropP}, for any $ \varepsilon > 0 $ there is some $ \varrho_\varepsilon \in \mathcal{D} $ such that 
$$
F \leq \inf_{x \in \mathcal{S}_2} \inf_{\pi \in \Pi_{\varrho_\varepsilon}}  \mathcal{P}\left(\pi, x \right) - \langle x , \varrho_\varepsilon  \rangle + \frac{1}{2} \zeta(\varrho_\varepsilon) + \varepsilon . 
$$

A matching lower bound is derived again by an approximation argument. Using the existing results  \cite[Thm.~1.2]{Chen23}, \cite[Prop 3.1]{Chen23b} (see also \cite{Pan18}) for vector-spin glasses, we conclude that
\begin{equation}\label{eq:Chen2}
\lim_{N\to \infty} \frac{1}{N} \mathbb{E}\left[ \ln Z_N^{D}   \right]   = \sup_{\varrho \in \mathcal{D}_D } \left[\inf_{x \in \mathcal{S}_{2,D}} \inf_{\pi \in \Pi_{\varrho, D}}  \mathcal{P}_D\left(\pi, x \right) - \langle x , \varrho  \rangle + \frac{1}{2} \zeta_D(\varrho)  \right],
\end{equation}
where  we use the notation~\eqref{eq:VGParisiF} from the proof of Theorem~\ref{thm:main}. The new ingredient is the set of paths
$$
 \Pi_{\varrho, D} \coloneqq \left\{ \pi: [0,1) \to \mathcal{D}_D  \; \mbox{c\`adl\`ag} \  \big| \  0 \leq \pi(s) \leq \pi(t) \; \mbox{for all $ s \leq t $ with }\; \pi(0) = 0  \; \mbox{and} \;  \pi(1) = \varrho \right\}  , 
$$
which end at $ \varrho \in \mathcal{D}_D $. Proceeding as in the proof Theorem~\ref{thm:main}, we conclude that the right side in~\eqref{eq:Chen2} is equal to 
\begin{multline*}
 \sup_{\varrho \in \mathcal{D}} \left[\inf_{x \in \mathcal{S}_{2}} \inf_{\pi \in \Pi_{\varrho}}  \mathcal{P}\left(Q_D\pi Q_D , Q_D x Q_D \right) - \langle Q_D x Q_D, Q_D \varrho  Q_D\rangle + \frac{1}{2} \zeta(Q_D\varrho Q_D)  \right] \\
 \geq  \inf_{x \in \mathcal{S}_{2}} \inf_{\pi \in \Pi_{\varrho_\varepsilon}}  \mathcal{P}\left(Q_D\pi Q_D , Q_D x Q_D \right) - \langle Q_D x Q_D,  \varrho_\varepsilon \rangle + \frac{1}{2} \zeta(Q_D\varrho_\varepsilon Q_D)  .
\end{multline*}
The lower limit $ D \to \infty $ is estimated from below using Lemma~\ref{lem:infconv} and $ \liminf_{D \to \infty} \zeta(Q_D\varrho_\varepsilon Q_D)  \geq \zeta(\varrho_\varepsilon ) $ by  the lower semicontinuity. We thus arrive at
$$
\liminf_{D\to \infty} \lim_{N\to \infty} \frac{1}{N} \mathbb{E}\left[ \ln Z_N^{D}   \right]  \geq  \inf_{x \in \mathcal{S}_2} \inf_{\pi \in \Pi_{\varrho_\varepsilon}}  \mathcal{P}\left(\pi, x \right) - \langle x , \varrho_\varepsilon  \rangle + \frac{1}{2} \zeta(\varrho_\varepsilon)  \geq F - \varepsilon ,
$$ 
which, using Proposition~\ref{prop:approx}, concludes the proof of a matching lower bound. 
\end{proof}
\appendix

\section{Hadamard-Schur product of certain integral operators}\label{app:HS}
In this appendix, we extend the well known properties of the Hadamard-Schur product of matrices~\cite{Bat97,Hiai:2009aa} to certain integral operators. 
We consider integral operators $ \varrho $  on the complex Hilbert space $ L^2(0,1) $  with a bounded, measurable  kernel $ \widehat \varrho : (0,1)^2\to \mathbb{C} $. Such integral operators are automatically Hilbert-Schmidt and hence bounded. 
Their Hadamard-Schur product $ \varrho \odot \pi $  is the bounded integral operator on $ f \in L^2(0,1) $:
\begin{equation}\label{eq:defSchurproduct}
	\left(\varrho \odot \pi \  f \right)(t) \coloneqq \int_0^1 \widehat  \varrho(t,s) \ \widehat \pi(t,s) \ f(s) ds .
\end{equation}

\subsection{Positivity and monotonicity}
We start by recording the analogue of the positivity of the Hadamard-Schur product of matrices. 
\begin{proposition}
If $ \varrho, \pi \geq 0 $ are non-negative integral operators on $ L^2(0,1) $ with bounded kernels, so is their Hadamard-Schur product $  \varrho \odot \pi \geq 0 $.
\end{proposition} 
\begin{proof}
The Hadamard-Schur product  is related to the tensor product $ \varrho \otimes \pi $ defined on $ L^2(0,1) \otimes L^2(0,1) $. More specifically, by the Lebesgue differentiation theorem
\begin{equation}\label{eq:tensor2rep}
\langle f ,  \varrho \odot \pi \   f \rangle = \lim_{r\downarrow 0 }  \ \langle f_r^{(2)} ,   \varrho \otimes \pi   f_r^{(2)} \rangle  . 
\end{equation}
where, for future purpose, we define for  $ k \geq 2 $ arbitrary:
\begin{equation}\label{eq:deffk}
 f^{k}_r(t_1,\dots,t_k)  \coloneqq    \frac{1}{(2r)^{k-1}} \prod_{l=2}^k 1[|t_1-t_l|< r] \  \prod_{j=1}^k f(t_j)^{1/k} .
\end{equation}
This function is in the $ k $-fold tensor product space $ L^2(0,1)^{\otimes k} $. 
The non-negativity thus immediately follows from that of the tensor product $ \varrho \otimes \pi $ and~\eqref{eq:tensor2rep}.
\end{proof} 
Iterating~\eqref{eq:defSchurproduct}, one defines the $k$-fold Hadamard-Schur product $ \varrho^{\odot k} $ on $ L^2(0,1) $ for arbitrary $ k \in \mathbb{N} $. 
If $ \varrho \geq 0 $, then $ \varrho^{\odot k} \geq 0 $ by the above proposition. 
We also need the monotonicity of this operation, which is again well known in the matrix case (cf.~\cite{Bat97,Hiai:2009aa}).
\begin{proposition}\label{pop:monotonicity}
If $ \varrho \geq \pi \geq 0 $ are non-negative integral operators on $ L^2(0,1) $ with bounded kernels, then $  \varrho^{\odot k }  \geq \pi^{\odot k } \geq 0 $ for any $ k \in \mathbb{N} $. 
\end{proposition}
\begin{proof}
A straightforward generalization of~\eqref{eq:tensor2rep} yields the following relation to the $ k $-fold tensor product $ \varrho^{ \otimes k} $:
$$
\langle f ,    \varrho^{\odot k }    f \rangle = \lim_{r\downarrow 0 }  \ \langle f_r^{(f)} ,   \varrho^{ \otimes k}   f_r^{(k)} \rangle  . 
$$ 
with  $ f_r^{(k)} $ from~\eqref{eq:deffk}. The claim follows from the fact that $ \varrho \geq \pi \geq 0 $ implies $  \varrho^{ \otimes k} \geq  \pi^{ \otimes k} \geq 0 $. 
\end{proof} 

\subsection{Properties of correlation functionals}\label{app:Lzeta}
In our context, we mostly encounter the expectation value 
\begin{equation}\label{eq:zetaexp}
\zeta_k(\varrho) \coloneqq  \langle 1 , \varrho^{\odot k} 1 \rangle = \int_0^1 \int_0^1 \widehat \varrho(t,s)^{k} dt ds  
\end{equation}
 of the  $ k $-fold Hadamard-Schur product of  a non-negative Hilbert Schmidt operator $ \varrho  \in \mathcal{S}_2^+ $ in the normalized vector $ 1 \in L^2(0,1) $. 
 Note that although initially only defined in case of bounded integral kernels $ \widehat \varrho $, the right side in~\eqref{eq:zetaexp} does make sense for any $  \varrho  \in \mathcal{S}_2 $ with $ \widehat \varrho \in L^{2}\left((0,1)^2\right) \cap L^k\left((0,1)^2\right) $. \\

  We start recording elementary properties of the above functional. 
 \begin{lemma} \label{lem:propzetak1}
On the space of non-negative Hilbert-Schmidt operators with a bounded integral kernel the functional $ \zeta_k $ is  for any $ k \in \mathbb{N} $:
\begin{enumerate}
\item non-negative and monotone, i.e.\ $0\ \leq \zeta_k(\pi) \leq \zeta_k(\varrho) $ for any $ 0 \leq \pi \leq \varrho $,
\item continuous, i.e.\ for any $ \varrho, \pi \in \mathcal{S}_2^+  $ with bounded integral kernels $  \widehat \varrho ,  \widehat \pi $ with $R\coloneqq \max\left\{ \|\widehat \varrho \|_\infty, \|\widehat \pi \|_\infty\right\} $:
\begin{equation}\label{eq:Lcontzetak}
\left| \zeta_k(\varrho)  - \zeta_k(\pi) \right| \leq k R^{k-1} \int_0^1\int_0^1 \left| \widehat \varrho(s,t) - \widehat\pi(s,t) \right| ds dt \leq   k R^{k-1} \left\| \varrho - \pi \right\|_2 . 
\end{equation}
\end{enumerate}
 \end{lemma}
 \begin{proof}
 The monotonicity is by Proposition~\ref{pop:monotonicity}. The continuity
 follows from the elementary inequality $ \left|x^k - y^k  \right| \leq k \max\{ |x|,|y|\}^{k-1} | x-y| $ for all $ x,y\in \mathbb{R} $.
 \end{proof}
 
 To further analyze the functionals $ \zeta_k(\varrho) $, it will be useful to recall the spectral representation of the intergral kernel,
 \begin{equation}\label{eq:spectral}
 \widehat\varrho(s,t) = \lim_{N\to\infty}  \widehat\varrho_N(s,t) , \quad  \widehat\varrho_N(s,t) \coloneqq \sum_{j=1}^N \lambda_j \varphi_j(s) \overline{\varphi_j(t)} ,
 \end{equation}
 in terms of the eigenvalues   $ (\lambda_j) \subset \mathbb{R} $  and the corresponding orthonormal  eigenfunctions  $ (\varphi_j) \subset L^2(0,1) $  of  $ \varrho \in \mathcal{S}_2 $. 
 The convergence in~\eqref{eq:spectral} is in $ L^2\left((0,1)^2\right) $, and hence pointwise for a.e.\ $(s,t) \in (0,1)^2 $ along a subsequence. 
As was already alluded to in the introduction, more is true in case $ \varrho \in \mathcal{S}_2^+ $ with a bounded kernel.
\begin{proposition}\label{lem:propzetak2}
For any $ \varrho \in \mathcal{S}_2^+ $ with a bounded kernel $ \widehat \varrho \in L^\infty\left((0,1)^2\right)  $, the integral kernel admits a bounded extension $ \widehat \varrho (t,t) $ to the diagonal given by~\eqref{eq:regkern} for which
\begin{enumerate}
\item for all $ N \in \mathbb{N} $ and a.e.\ $t \in  (0,1) $: \quad
$\displaystyle
 \widehat\varrho_N(t,t)=  \sum_{j=1}^N \lambda_j |\varphi_j(t)|^2 \leq  \widehat \varrho (t,t) $.\\
Consequently, the eigenfunctions $ (\varphi_j)  $ corresponding to non-zero eigenvalues are bounded. 
\item for a.e.\ $(s,t) \in (0,1)^2 $: \quad $\displaystyle   |  \widehat\varrho(s,t)  | \leq \sqrt{ \widehat \varrho(s,s) \ \widehat\varrho(t,t) } $, \quad and  $ \tr \varrho = \int  \widehat \varrho (t,t) dt $. 
\item 
 for any $ k \in \mathbb{N} $:\quad 
 $\displaystyle
 \zeta_k(\varrho) =  \int_0^1 \int_0^1 \widehat{ \varrho}(t,s)^{k} \ dt ds = \sum_{j_1,\dots,j_{k}=1}^\infty \lambda_{j_1}\cdots  \lambda_{j_{k}} \left|\langle 1 ,  \varphi_{j_1}\cdots \varphi_{j_{k}} \rangle\right|^2 $.
 \end{enumerate}
\end{proposition}
\begin{proof} 
The proofs of 1.~and 2.~can be found in~\cite{Weid66}. Note that the bound in~1.~indeed guarantees
that $ \lambda_j |\varphi_j(t)|^2 \leq \widehat\varrho(t,t) $ for a.e.\ $ t \in (0,1) $ and all $ j \in \mathbb{N} $. Consequently, the scalar products
$$
\langle 1 ,  \varphi_{j_1}\cdots \varphi_{j_{k}=1} \rangle = \int_0^1  \varphi_{j_1}(t) \cdots \varphi_{j_{k}}(t) dt 
$$
are all finite. 
We now approximate $ \varrho $ by truncating its spectral representation at $ N \in \mathbb{N} $. Let $ \varrho_N $ be the finite-rank operator with integral kernel given by~\eqref{eq:spectral}.
By construction $0\leq  \varrho_N \leq  \varrho $ and $\lim_{N\to\infty}  \|  \varrho_N -  \varrho \|_2 = 0 $. Consequently, using~\eqref{eq:Lcontzetak} and the bound in 2.~on the kernels, 
$$
\lim_{N\to\infty} \zeta_k(\varrho_N) =  \zeta_k(\varrho) , \qquad 0 \leq  \zeta_k(\varrho_N) \leq \zeta_k(\varrho) ,
$$
where  the last inequality is by Proposition~\ref{pop:monotonicity}. 
By the finiteness of the scalar products, an elementary computation based on an application of Fubini's theorem yields for any $ N \in \mathbb{N} $
\begin{equation}\label{eq:finiteNcomp}
 \zeta_k(\varrho_N) =   \sum_{j_1,\dots,j_{k}=1}^N \lambda_{j_1}\cdots  \lambda_{j_{k}} \left|\langle 1 ,  \varphi_{j_1}\cdots \varphi_{j_{k}} \rangle\right|^2 .
\end{equation}
By monotone convergence, the right side converges to the right side in the assertion as $ N \to \infty $. This completes the proof. 
\end{proof}

Our final result concerns the extension of the functional to $ \mathcal{S}_2 $:
$$ \widetilde\zeta_{2p}: \mathcal{S}_2 \to [0,\infty] , \quad \widetilde\zeta_{2p}(\varrho) \coloneqq \int_0^1\int_0^1 \widehat \varrho(s,t)^{2p} ds dt  .
$$
Note that the integral is well-defined, potentially infinite, for all $ \widehat \varrho \in L^2\left( (0,1)^2\right) $. 
\begin{proposition}\label{prop:extensionzeta}
The functional $  \widetilde\zeta_{2p}: \mathcal{S}_2 \to [0,\infty] $ is weakly lower semicontinuous for any $ p \in \mathbb{N} $. \end{proposition} 
\begin{proof}
We pick $ \varrho_n , \varrho \in \mathcal{S}_2 $ such that $ \lim_{n\to \infty} \langle x,  \varrho_n \rangle = \langle x, \varrho \rangle  $ for all $ x \in \mathcal{S}_2 $. 
Let $ R > 0 $ be arbitrary, and pick $ x_R \in \mathcal{S}_2 $ with integral kernel
$$
\widehat x_R(s,t) \coloneqq  e^{i \arg \widehat \varrho(s,t) } \left|  \widehat \varrho(s,t) \right|^{2(p-1)} \indfct[ \left|  \widehat \varrho(s,t) \right| \leq R ]  .
$$
From weak convergence we conclude
$$
I_R(2p)\coloneqq \int_0^1\int_0^1 \left|  \widehat \varrho(s,t) \right|^{2p} \indfct[ \left|  \widehat \varrho(s,t) \right| \leq R ]  ds dt = \lim_{n\to \infty} \langle x_R ,  \varrho_n \rangle \leq   \liminf_{n\to \infty} \| \widehat \varrho_n \|_{2p}  \ I_R(2p)^{1-\frac{1}{2p}} ,
$$ 
where the last step is H\"older's inequality. Hence, $  \liminf_{n\to \infty}   \| \widehat \varrho_n \|_{2p}^{2p}  \geq I_R(2p) $ for any $ p \geq 1 $ and $ R > 0 $. Combined with Fatou's lemma, this yields:
$$
\liminf_{n\to \infty}  \widetilde\zeta_{2p}(\varrho_n) =  \liminf_{n\to \infty}   \| \widehat \varrho_n \|_{2p}^{2p} \geq  I_R(2p) . 
$$
As $ R \to \infty $,  the right side converges to $ \zeta(\varrho) $  by monotone convergence. This finishes the proof of the weak lower semicontinuity.
\end{proof}
%
%\subsection{Proofs of Lemma~\ref{lem:zetaprop} and Lemma~\ref{lem:thetaprop}}
%
%\begin{proof}[Proof of Lemma~\ref{lem:zetaprop}] 
%1.~The convexity immediately follows from the convexity of the functions $ [0,\infty) \ni x \mapsto x^{2p} $ for any $ p \in \mathbb{N} $. 
%The weak lower semicontinuity follows from Proposition~\ref{prop:extensionzeta}. 
%The claimed bound on $\zeta $ for $ \varrho \in \mathcal{D} $  is immediate from the boundedness of the kernel  $| \widehat \varrho |\leq 1 $. \\
%\noindent
%2.~This follows straightforwardly from Lemma~\ref{lem:propzetak1}.\\
%\noindent
%3.~The Gateaux differentiability and its representation in terms of the operator~\eqref{eq:covariance form} is immediate. The boundedness with the explicit norm bounds of the operator $ d\zeta(\varrho) $ 
%are also straightforward from its definition as a Hadamard-Schur product. So is its monotonicity by Lemma~\ref{lem:propzetak1}. 
%The proof of~\eqref{eq:Lcdzeta} follows from the same elementary inequality as~\eqref{eq:Lcontzetak}. 
%\end{proof}
%
%
%
%\begin{proof}[Proof of Lemma~\ref{lem:thetaprop}] 
%1.~The convexity again follows from the convexity of the functions $ [0,\infty) \ni x \mapsto x^{2p} $ for any $ p \in \mathbb{N} $. The uniform boundedness is immediate from the boundedness of the kernel  $| \widehat \varrho |\leq 1 $.\\
%\noindent
%2.~The representation~\eqref{eq:thetadef} as a non-negative series involving the Hadamard product and Proposition~\ref{pop:monotonicity} yield the claim. \\
% 3.~The proof of the Lipschitz continuity parallels the proof of~\eqref{eq:Lcdzeta}. 
%\end{proof} 

\section{Properties of the quantum Parisi functional}\label{app:PF}

In this appendix, we collect the proofs of properties of the quantum Parisi functional listed in Subsection~\ref{sec:QPF}. 
\subsection{Continuity properties }
\begin{proof}[Proof of Proposition~\ref{lem:LcontP}] 
The claim~\eqref{eq:LcontP} follows by establishing separately the Lipschitz continuity of the functional $  X_0(\pi,x) $ and $ \int_0^1 \theta(\pi(s)) ds $.  The latter immediately follows from~\eqref{eq:thetaLc}. It thus remains to establish the Lipschitz continuity of  $  X_0(\pi,x) $. For its proof we follow closely~\cite{AdBr20}, which in turn is based on the Gaussian interpolation argument in~\cite{PanPotts18}. \\

Consider two  $ \pi , \pi' \in \Pi $, and suppose without loss of generality that they correspond to a common (refined) partition $ 0 =: m_0 < m_1 < \dots < m_r < m_{r+1}:= 1$. Let $  (\mu^{\pi,\pi'}_\alpha)_{\alpha \in \mathbb{N}^{r} }$ stand for the corresponding weights and consider two independent families of $ L^2 $-valued, centred Gaussian processes $ (z_\alpha)_{\alpha \in \mathbb{N}^{r} }$ and $  (z'_\alpha)_{\alpha \in \mathbb{N}^{r} }$ with covariances given by \eqref{eq:defz} with values $ \pi(m_{j}) $ and $ \pi'(m_{j}) $ respectively.  The interpolated Gaussian process 
$$
z_\alpha(r) := \sqrt{r}z_\alpha + \sqrt{1-r} z'_\alpha  , \quad r \in [0,1], 
$$
then has the covariance 
$$
\mathbb{E}\left[ \langle f , z_\alpha(r) \rangle \langle g , z_{\alpha'}(r) \rangle \right]  =  r \langle f , d\zeta\left(\pi(m_{\alpha \wedge \alpha'} )  \right)g \rangle + (1-r) \langle f , d\zeta\left(\pi'(m_{\alpha \wedge \alpha'} )  \right)g \rangle , \quad f,g \in L^2(0,1) . 
$$ 
The main idea is to control the derivative of
\begin{align*}
& f(r) := \\
& \mathbb{E}\left[ \ln \sum_{\alpha \in \mathbb{N}^{r} } \mu^{\pi,\pi'}_\alpha\! \int \exp\left( \langle \xi ,  z_\alpha(r)  \rangle + r \langle \xi , \left( x- \tfrac{1}{2} d\zeta(\pi(1) \right)  \xi  \rangle + (1-r) \langle \xi , \left( x'- \tfrac{1}{2} d\zeta(\pi'(1) \right)  \xi  \rangle\right) \nu_1(d\xi) \right] ,
\end{align*}
which coincides with $ X_0(\pi,x) $ at $ r = 1 $ and with $ X_0(\pi',x') $ at $ r = 0 $. 
Let $ \langle (\cdot) \rangle_r $ denote the (random) Gibbs expectation value associated to partition function defined by the logarithm, and $ \langle (\cdot) \rangle_r^{\otimes 2} $ the Gibbs expectation value of the corresponding duplicated system. 
Then straightforward differentiation and Gaussian integration by parts, yields
\begin{equation}\label{eq:derPf}
f'(r) = \frac{1}{2}\  \mathbb{E}\left[ \left\langle  \langle \xi ,  d\zeta\left(\pi'(m_{\alpha\wedge\widetilde\alpha})\right) \widetilde \xi \rangle \right\rangle_r^{\otimes 2} -   \left\langle  \langle \xi ,  d\zeta\left(\pi(m_{\alpha\wedge\widetilde\alpha})\right) \widetilde \xi \rangle \right\rangle_r^{\otimes 2}\right]  +  \mathbb{E}\left[ \left\langle  \langle \xi , (x-x')  \xi \rangle \right\rangle_r\right] ,
\end{equation}
where the variables with and without the tilde refer to the first and second system in the duplication. 
Note that two terms $  \frac{1}{2}\  \mathbb{E}\left[ \left\langle  \langle \xi ,  d\zeta\left(\pi(1)\right) \xi \rangle \right\rangle_r \right] $ and $ \frac{1}{2}\  \mathbb{E}\left[ \left\langle  \langle \xi ,  d\zeta\left(\pi'(1)\right) \xi \rangle \right\rangle_r \right] $, which appear in the Gaussian integration by parts cancels due to the presence of this term in the exponent. This implies that
$$
\left| f'(r) \right| \leq \frac{1}{2}   \mathbb{E}\left[\left\langle \left\| d\zeta\left(\pi(m_{\alpha\wedge\widetilde\alpha})\right) - d\zeta\left(\pi'(m_{\alpha\wedge\widetilde\alpha})\right) \right\| \right\rangle_r \right] + \| x - x' \| .
$$

The weights of the Ruelle probability cascade have the property \cite[Thm.~4.4]{Pan13} that
\begin{align*}
& \mathbb{E}\left[\left\langle \left\| d\zeta\left(\pi(m_{\alpha\wedge\widetilde\alpha})\right) - d\zeta\left(\pi'(m_{\alpha\wedge\widetilde\alpha})\right) \right\| \right\rangle_r \right]  = \sum_{\alpha, \tilde\alpha}  \mathbb{E}\left[ \mu^{\pi,\pi'}_\alpha \mu^{\pi,\pi'}_{\tilde\alpha} \right] \left\| d\zeta\left(\pi(m_{\alpha\wedge\widetilde\alpha})\right) - d\zeta\left(\pi'(m_{\alpha\wedge\widetilde\alpha})\right) \right\| \\
 & = \sum_{j=0}^{r+1} \left\| d\zeta\left(\pi(m_j)\right) - d\zeta\left(\pi'(m_j)\right) \right\|  \sum_{\alpha\wedge\widetilde\alpha = j}  \mathbb{E}\left[ \mu^{\pi,\pi'}_\alpha \mu^{\pi,\pi'}_{\tilde\alpha} \right]  =  \int_0^1 \left\| d\zeta\left(\pi(m)\right) - d\zeta\left(\pi'(m)\right) \right\| dm \\
 & \leq L_\zeta  \int_0^1 \left\|\pi(m) - \pi'(m)\right\|_2 dm ,
 \end{align*}
where the last step is by~\eqref{eq:Lcdzeta}. 
This completes the proof of~\eqref{eq:LcontP}.\\

\end{proof}

\subsection{Convexity, bounds and differentiability}

\begin{proof}[Proof of Proposition~\ref{prop:elpropP}] 
1.~The convexity of $ x \mapsto \mathcal{P}(\pi , x ) $ is immediate from that of $ X_0(\pi,x)  $ and the representation~\eqref{eq:PRuelle}. So is the boundedness of $ \int_0^1\theta(\pi(m)) dm $ from Lemma~\ref{lem:thetaprop}. 

The lower bound on $ X_0(\pi,x) $ for any $ \pi \in \Pi $ follows from an application of Jensen's inequality to~\eqref{eq:PRuelle}, which yields
$$
 X_0(\pi,x)\geq \mathbb{E}\left[  \sum_{\alpha \in \mathbb{N}^{r} } \mu^\pi_\alpha \int  \langle \xi ,  z_\alpha \rangle + \langle \xi , \left( x- \tfrac{1}{2} d\zeta(\pi(1) \right)  \xi  \rangle  \nu_1(d\xi) \right] \geq -   \int  \langle \xi , \left( x_- + \tfrac{1}{2} d\zeta(\pi(1) \right)  \xi  \rangle \ \nu_1(d\xi) .
$$
The scalar product is estimated either by $  \langle \xi , x_- \xi \rangle \leq \| x_- \| $, or, using its integral representation and the boundedness of the paths, by $  \langle \xi , x_- \xi \rangle \leq \int_0^1 \int_0^1 |\widehat x_-(s,t)| ds dt $. 
The operator norm bound $ \| d\zeta(\pi(1)) \| \leq N_\zeta $ is from Lemma~\ref{lem:zetaprop}.  

The upper bound on $ X_0(\pi,x) $ is also by Jensen's inequality, which using the fact that the exponential contains a martingale, yields 
\begin{equation}\label{eq:annealedJensen}
 X_0(\pi,x)  \leq \ln  \int \exp\left( \langle \xi ,  x  \xi  \rangle \right) \nu_1(d\xi)  \leq  \ln  \int \exp\left( \langle \xi ,  x_+ \xi  \rangle \right) \nu_1(d\xi) .
\end{equation}
The above bounds on the scalar product yield the claim concerning $ X_0(\pi,x)  $. 
%Since $ \mathcal{P}(\pi,x) \geq X_0(\pi,x)  $, the lower bound on $ \inf_{\pi \in \Pi}  \mathcal{P}(\pi,x)  $ immediately follows. The upper bound is a conequence of $  \inf_{\pi \in \Pi}  \mathcal{P}(\pi,x) \leq  \mathcal{P}(0,x) = X_0(0,x) $. \\

An explicit computation starting from~\eqref{eq:PRuelle} yields for any $ x , y \in \mathcal{S}_2 $
$$
 \frac{d}{d\lambda} \mathcal{P}(\pi, x+ \lambda y ) \Big|_{\lambda = 0 } = \mathbb{E}\left[ \left\langle \langle r_1(\xi, \xi) , y \rangle\right\rangle_{x} \right] ,
$$
where $ \langle (\cdot )\rangle_{x} $ abbreviates the average corresponding to $  \sum_{\alpha \in \mathbb{N}^{r} } \mu^\pi_\alpha \int \exp\left( \langle \xi ,  z_\alpha \rangle + \langle \xi , \left( x- \tfrac{1}{2} d\zeta(\pi(1) \right)  \xi  \rangle \right) \nu_1(d\xi)  $.  For an estimate, we use  $ \|  r_1(\xi, \xi)  \|_2 = 1  $. 
Differentiating again then yields 
\begin{equation}\label{eq:secondd}
 \frac{d^2}{d\lambda^2} \mathcal{P}(\pi, x+ \lambda y ) \Big|_{\lambda = 0 } =  \mathbb{E}\left[ \left\langle \langle r_1(\xi, \xi) , y \rangle^2\right\rangle_{x} - \left\langle  \langle \langle r_1(\xi, \xi) , y \rangle\right\rangle_{x}^2  \right] ,
\end{equation}
from which the bound immediately follows.\\

\noindent
2.~The first inequality is~\eqref{eq:annealedJensen} combined with the fact that $ \inf_{\pi \in \Pi} \int_0^1 \theta(\pi(m)) dm = 0 $.  
The next inequality results from 1.~and the trivial estimate of the supremum by $ \inf_{\pi \in \Pi}  \mathcal{P}(\pi,0)  - \frac{1}{2} \zeta^*(0) = \inf_{\pi \in \Pi}  \mathcal{P}(\pi,0) \geq - N_\zeta/2 $ by 1. The third inequality is a simple consequence of the first.  The last inequality used the fact that by the Fenchel-Young inequality for conjugate pairs \cite[Prop. 13.13]{BauCom17} for any  $ \xi \in L^2 $ with corresponding rank-one operator $ | \xi\rangle\langle \xi | \in \mathcal{D} $:
$$
\langle \xi ,  x \xi  \rangle \leq \frac{1}{2} \zeta^*(2x) + \frac{1}{2} \zeta(| \xi\rangle\langle \xi |  ) \leq  \frac{1}{2} \zeta^*(2x) + \frac{1}{2} \widehat\zeta(1),
$$
where the last step is by Lemma~\ref{lem:zetaprop}. 
\end{proof}

\minisec{Acknowledgments}
This work was supported by the DFG under EXC-2111--390814868, and to a minor extend also under DFG--TRR 352--Project-ID 470903074.
SW would like to thank Silvio Franz for helpful discussions. 

\bibliography{QParisi}{}
\bibliographystyle{plain}

	\bigskip
	\bigskip
	\begin{minipage}{0.5\linewidth}
		\noindent Chokri Manai and Simone Warzel\\
		Departments of Mathematics \& Physics, and MCQST \\
		Technische Universit\"{a}t M\"{u}nchen
	\end{minipage}

\end{document}